\documentclass[11pt, letterpaper]{article}

\usepackage{titletoc}
\usepackage[toc, page, header]{appendix} 

\usepackage[margin=1in]{geometry}
\usepackage{amsmath, amssymb, amsthm,color}
\usepackage{thmtools}
\usepackage{graphicx}
\usepackage[colorlinks=true, linkcolor=blue!60!black, citecolor=blue!60!black,
urlcolor=black]{hyperref}
\usepackage{xcolor}
\usepackage{titlesec}
\usepackage{enumitem}
\usepackage[semicolon,round]{natbib}
\usepackage[normalem]{ulem}
\usepackage{bm}
\usepackage{xspace}

\usepackage{fancyhdr}
\usepackage{algorithm}
\usepackage[noEnd,commentColor=blue!70!white]{algpseudocodex}

\makeatletter
\newcommand{\algeq}[2][]{%
  \par\vspace{5pt}%
  \noindent%
  \makebox[\columnwidth]{%
    \hfill #2%
    \if\relax\detokenize{#1}\relax
      \hfill
    \else
      \hfill (\refstepcounter{equation}\theequation)\label{#1}%
    \fi
  }%
  \par\vspace{5pt}%
}
\makeatother

\usepackage[noabbrev,sort]{cleveref}
\AddToHook{cmd/appendix/before}{%
    \crefalias{section}{appendix}%
    \crefalias{subsection}{appendix}
}

\Crefname{ALC@line}{Line}{Lines}
\Crefname{mechanism}{Mechanism}{Mechanisms}

\newtheorem{theorem}{Theorem}
\newtheorem{proposition}[theorem]{Proposition}
\newtheorem{lemma}[theorem]{Lemma}

\theoremstyle{definition}
\newtheorem{definition}{Definition}

\newtheorem{remark}[definition]{Remark}

\usepackage{tikz}  \usetikzlibrary{shapes, arrows.meta, positioning, fit, backgrounds} 

\usepackage{comment}
\newcommand{\E}{\operatornamewithlimits{\mathbb{E}}}
\newcommand{\argmax}{\operatornamewithlimits{\text{argmax}}}
\newcommand{\Var}{\operatorname{Var}}
\renewcommand{\paragraph}[1]{\noindent\textbf{#1}}

\newcommand{\gai}{AI\xspace}

\title{Market Design for \gai: Beyond the Copyright Binary
}
\author{Yan Dai\footnote{MIT Operations Research Center. \texttt{yandai20@mit.edu}.} \qquad Maryam Farboodi\footnote{Cornell, NBER and CEPR. \texttt{m.farboodi@gmail.com}.} \qquad Negin Golrezaei\footnote{MIT Sloan School of Management. \texttt{golrezae@mit.edu}.} \qquad Sepehr Shahshahani\footnote{Washington University School of Law. \texttt{sepehrs@wustl.edu}.}}
\date{First version: Feb 2026. This version: Jul 2026.\footnote{We thank Xavier Gabaix, Mark Lemley, Ilan Morgenstern, Lisa Larrimore Oullette, Eva Tardos, and Yifan Wu, as well as the participants at Wharton Accountable AI, Stanford Market Design in the Age of AI, Marketplace Innovation Workshop, IP Researchers Europe, Society for Economic Research on Copyright Issues, Society for Institutional and Organizational Economics, ACM Economics and Computation (EC), and INFORMS M\&SOM conferences for helpful conversations and comments.}}

\begin{document}

\maketitle

\begin{abstract}
\noindent 
How can we design a market of human-generated content for use in training AI models that both enables technological progress and preserves individual incentives for high-quality content creation?
Existing approaches take polar positions: a “free-for-all” model based on fair use and a “strong intellectual property rights” model. 
We show that both fail: Free-for-all does not compensate creators, and---by modeling as a static Stackelberg game---strong intellectual property rights also underpower creative incentives. We find this especially true for more innovative creators, a phenomenon we term the “originality penalty.”
Extending this insight to a dynamic model, we find another market failure undermining AI model performance, even for an initially good model: Such a model induces greater reliance by humans on AI-assisted creation, resulting in homogenized content feeding back into training, which degrades the model performance---a “curse of precision.” 
We further propose a market design with a data intermediary { negotiating collectively with the AI firm} and subsidizing innovative contributions, thus restoring efficiency. 
\end{abstract}

\section{Introduction}

Recent breakthroughs in generative artificial intelligence (\gai ) are laden with both promise and peril. 
By enabling the fast, cheap delivery of useful content, \gai  models have improved work and life for millions of people. 
But \gai models' ability to generate this useful content depends on a rich reservoir of human-created content on which to train. 
To date, AI firms have acquired this content by scraping it from online sources—largely without human creators' consent, compensation, or credit.
While this practice has enabled rapid \gai progress, it undermines human creators' incentives to create.
We are already seeing warning signs: Stack Overflow’s traffic declined significantly after the rise of Large Language Models (LLMs) like ChatGPT, as users consumed \gai answers rather than contributing to producing them \citep{burtch2024consequences,del2024large,sharma2025beyond}. 
{ A recent field experiment documents the same dynamic for Google's AI Overviews, which divert traffic away from publishers whose content is summarized \citep{agarwal2026google}.}
Various artists and content creators have sued AI firms contending that their practices have deprived them of the rewards of their creative work (\Cref{sec:literature}).
These trends point to a misalignment between the incentives of \gai developers and human creators, and they give rise to the central research question of this project: 
How can we design a market of human-generated content for \gai that both enables technological progress and preserves individuals' incentive to create high-quality content?

This is among the foremost policy challenges of our times, implicating in its narrow berth important questions of innovation and technology policy and, more broadly, fundamental concerns about how to preserve human creative flourishing.
But the tools we have come from a decades-old Copyright Act that was not designed with anything like \gai  technology in mind.
As such, the current debate is trapped in a legal binary: \gai  harvesting is either fair use (what we call a ``free for all'' system) or copyright infringement (a ``strict intellectual property rights'' regime). 
This paper argues that both extremes are destined to fail and points instead to a ``third way'' via data intermediaries.

The reason a free-for-all system fails is easy to grasp: 
If AI can freely take human-created content and use it to generate content that largely replaces consumers' need for the human-created content, then humans will be left with much-diminished incentives to invest in high-quality content creation. 
But we show that the opposite solution---strict intellectual property rights---is also doomed to fail.

The reasons for this failure are subtle and unexpected. We begin with a static Stackelberg model capturing the interaction between human creators and the AI firm.
First, contrary to conventional wisdom about the effects of intellectual property (IP) rights, we show that {strong individualistic IP rights result in under-powered creative incentives}. 
{ This is driven by the market power of AI firms: a buyer with monopsony power can profitably price content below what would be required to induce socially optimal effort on the part of human creators. 
Correlation among different human creators' content---the fact that, by selling her own content, a creator is effectively selling ``part of'' other creators'---does not by itself widen this shortfall, but, as we discuss next, it determines which creators bear it.} 
Second---and more surprisingly---we show that AI firm's preference for representative data creates an ``originality penalty,'' meaning the market degrades the incentives of innovative content creators to a greater extent (relative to social optimum) than it degrades the incentives of run-of-the-mill creators who are highly influenced by common trends and biases. 
This originality penalty thus leads to the production of homogenized content.
Although it is predictable that AI firms' market power may lead to under-investment by human creators, our finding of an originality penalty contradicts the economic intuition that scarcity demands a premium, revealing a market failure that is distinctive to the \gai context. 

Beyond the short-run tradeoff between human creative incentives and technological progress, we further extend our analysis to a continuous-time dynamic setting in which an AI firm trains and refines its models on human-generated content over time, and humans in turn use AI tools to assist them generating content.
In this dynamic setting, we identify a ``curse of precision'' whereby a good model induces greater reliance by human content creators on \gai outputs, leading to a homogenization of human content, which feed back into the \gai training pipeline, degrading model performance.
The unexpected consequence is that a good model sows the seeds of its own decay.
While this ``curse of precision'' looks similar to the phenomenon of ``model collapse'' identified in the empirical machine learning literature \citep{alemohammad2023self,shumailov2023curse}---when trained on self-generated data, an \gai model's capability collapses---our analysis, however, reveals a different pathway to failure and a more pessimistic prediction: 
Even if the AI firm constantly acquires {fresh data} from human creators, poorly designed markets can lead to \gai model failure.

Having identified the static and dynamic properties of market equilibrium under existing institutional designs, we then move beyond diagnosis to treatment by proposing a better market design. 
Our proposed solution takes the form of a {data intermediary} who 
negotiates on behalf of individual content owners with an AI firm and apportions royalties among them. 
We show that, given full information, the data intermediary can overcome both market failures that characterize the individual-IP-rights regime.  
By managing creators as a portfolio and negotiating collectively on their behalf, the intermediary overcomes the AI firm's market power, { securing creators a share of the social surplus their content generates}. 
And by apportioning the payment received from the AI firm among individual creators according to their respective contributions to social surplus (using weights that we define inspired by the Aumann--Shapley value), the intermediary compensates the most original creators in proportion to their contributions---rather than marking them down the most, as the unintermediated market does through the originality penalty. 
We show, interestingly, that efficiency can be restored using a simple {two-part tariff} mechanism where both the lump-sum transfer that the intermediary receives from the AI firm and the amounts that the intermediary pays to individual content creators are affine functions of the effort exerted by content creators. 

Our intermediary solution abstracts away from informational asymmetries and, more generally, from principal-agent problems between the intermediary and human creators. 
As such, it should not be viewed as a definitive policy solution. 
The idea is nevertheless quite useful as a first step toward an optimal market design because any problems that characterize the full-information context will also plague an incomplete-information environment. 
The market failures we identify will exist even in the absence of information gaps or principal-agent costs, and they are the very failures addressed by our proposed market design, so the solution is well-calibrated to the problem as defined and serves as a useful benchmark for a definitive solution. 

The paper is organized as follows.
The next section explains the legal-policy background and situates our contribution in the literature. 
\Cref{sec:model} develops and solves a static model that uncovers the shortcomings of the existing legal regime and commonly proposed solutions. \Cref{sec:dynamic} extends the model to a continuous-time dynamic case, identifying a long-run market failure. 
Building on these results, \Cref{sec:mechanism} proposes an institutional design that overcomes these shortcomings and outperforms commonly proposed solutions. 
The Appendix furnishes formal statements and proofs.

\section{Policy Debate and Literature}\label{sec:literature}

The question of how best to resolve AI's disruption of the market for creating and consuming informational content is hotly debated not just in academic literature but also in courts {and among legislators \citep{lucchi2026impact}}.
The most consequential legal questions affecting the future of \gai  have been raised in several lawsuits filed by content owners against AI firms alleging copyright infringement.\footnote{A useful list appears at \href{https://chatgptiseatingtheworld.com/2024/08/27/master-list-of-lawsuits-v-ai-chatgpt-openai-microsoft-meta-midjourney-other-ai-cos/}{https://chatgptiseatingtheworld.com/2024/08/27/master-list-of-lawsuits-v-ai-chatgpt-openai-microsoft-meta-midjourney-other-ai-cos/}.
For a non-lawyer-friendly introduction, see \citet{samuelson2023,samuelson2025}.
}  
To a layperson, the juxtaposition of copyright with the cutting edge technology may appear unusual---after all, copyright is thought to be about artistic creativity, whereas the high-technology context might seem better suited to patent law or other regulatory regimes---but in fact copyright law has long been at the forefront of regulating new technologies in the United States. 
From the piano roll to cable to VCRs to peer-to-peer file sharing to online video streaming and beyond, the first steps in deciding the fate of emerging technologies have often been taken by courts in copyright cases \citep{shahshahani2018, samuelson2024}. 
To appreciate the reasons for this, it is essential to review the fundamental structure of American IP laws, of which copyright is a central part. 

The goal of American IP laws, in the words of the U.S. Constitution, is “to promote the progress of science and useful arts, by securing for limited times to authors and inventors the exclusive right to
their respective writings and discoveries” (U.S. Const. art. I, \S\,8, cl. 8). 
The reason such an ``exclusive right''---in other words, a temporary monopoly---has been considered necessary to achieve the objective of progress in arts and science traces to the ``public good'' nature of intangible products of the mind: 
Unlike tangible goods (say, a chair), intangible products are ``nonrivalrous''---one person's use does not diminish another person's use---and ``nonexclusive''---the use of an idea once disclosed is hard to limit to the first recipient. 
As such, the theory of American IP laws is that exclusive rights are necessary to prevent freeriding on others' creativity and thus maintain incentives to create.  
But the theory also acknowledges that these exclusive rights come at a cost: 
Like any monopoly, IP rights raise the cost of a work protected by IP and make it more difficult for users to access. 
These include not only costs to end users, such as a person who wants to read a copyrighted novel or take a patented medicine, but also---of special importance in the \gai  context---costs to follow-on creators who want to use existing content as building blocks for creating new content \citep{freilichshahshahani2023, shahshahani2025}. 
Balancing creative incentives against access costs---the {incentives-access tradeoff}---is thus at the core of American intellectual property law. 
This has been the dominant conceptual framework for understanding IP in economics for a long time \citep[e.g.,][]{plant1934,machlup1958,arrow1962economic,nordhaus1969} and in US law for an even longer time (e.g., \textit{Wheaton v. Peters}, 33 U.S. 591, 657–58, 661 (1834); \textit{Kendall v. Winsor}, 62 U.S. 322, 327–29 (1858)); \textit{Sears, Roebuck \& Co. v. Stiffel Co.}, 376 U.S. 225, 229–31 (1964); \textit{Twentieth Century Music Corp. v. Aiken}, 422 U.S. 151, 156 (1975)).

The incentives-versus-access paradigm helps explain the current array of scholarly opinion on the question of whether the unauthorized use of copyrighted content in training \gai  models should be permitted.\footnote{
That is not the only copyright-law question implicated by \gai. 
Grossly simplified, there are least three major legal questions: 
(1) Does the use of copyrighted material in an \gai  system’s training set infringe the copyright in the material?  
(2)	Does the output of an \gai model infringe the copyright in works in its training set or other works?  
In this connection, does the substantial-similarity standard differ from the standard that is applied when considering non-\gai cases of copyright infringement? 
(3) Is output produced by someone using \gai copyrightable? 
If so, who is the author (initial owner)? 
It is the first of these questions that is most directly applicable to this paper, though the second question also potentially comes into play.
For analysis of these and other questions related to copyright and \gai, see 
\cite{henderson2023, lemley2024, lemleyoullette2025}.
} 
On one side, scholars have argued that AI firms' unrestricted use of copyrighted works in training \gai models undermines individuals' incentives to create high-quality content, so robust IP rights are needed to preserve creative incentives and promote 
 the constitutional goal of artistic and scientific progress \citep{opderbeck2024,barnett2025}.
 We call this position the ``strict intellectual property rights'' model.
On the other side, scholars have argued that unrestricted access to data is essential to innovation, concluding that the use of copyrighted works in \gai training constitutes ``fair use''\footnote{Fair use is a longstanding judge-made doctrine, now codified in the Copyright Act, 17 U.S.C. \S\,107, which holds that certain uses of copyrighted works, though otherwise impinging on the exclusive rights of the copyright owner, qualify as noninfringing in light of the purpose and character of the use, the nature of the copyrighted work, the amount and substantiality of the material used, and the effect of the use on the market for the original work. 
The two main theoretical justifications for the fair use doctrine are that it cures market failure in situations where the copyright owner is unlikely to license a valuable use \citep{gordon1982} and that it encourages ``transformative'' uses whereby a second-generation creator uses a copyrighted work as raw material in the creation of new information, insights, or aesthetics \citep{leval1990}.  
While these academic justifications are useful and influential, the doctrine is highly fact-dependent and case results cannot be confidently predicted based only on academic theory.} and is not a copyright infringement \citep{sag2019new,lemleycasey2021,lee2025}. 
We call this the ``free for all'' model. 

Our first insight is that these polar ways of thinking about the problem fall short. 
It is easy to see why the free-for-all model will not work: 
If content can be taken without rewarding its creator, people will have little incentive to invest resources into creating content, so the amount and quality of human-created content will decrease. 
That is the very freerider problem that American copyright and patent laws are designed to ameliorate. 
The problem will be especially pronounced as consumers increasingly turn to \gai platforms, rather than the original source, to get the content they want.

Admittedly, this analysis does not do full justice to the concept of fair use, which is flexible and capable of adaptation to a variety of factual circumstances.
The doctrine does not require that {all} uses of copyrighted work for training \gai models be deemed noninfringing (though some commentators have advocated just that); depending on the circumstances, some uses may be shielded by fair use while others are held to be copyright infringement. 
For example, in a case involving Anthropic, the firm that developed the large language model Claude, a court recently held that the unauthorized copying and use of copyrighted books to train the model qualified as fair use with respect to books that Anthropic purchased but not with respect to books that it downloaded for free from ``pirate'' websites (\textit{Bartz v. Anthropic PBC}, 787 F. Supp. 3d 1007, 1033 (N.D. Cal. 2025)). 
Whether particular \gai uses qualify as fair use could thus be determined case by case in light of particular factual circumstances, as the recent Copyright Office Report on artificial intelligence predicts \citep[74]{copyrightoffice2025}. 
While such flexibility can be a virtue, it also breeds uncertainty and unpredictability, a longstanding criticism of fair use doctrine (see \cite{shahshahani2015} for a critical survey). 
For example, on facts very similar to the \textit{Anthropic} case, a different court held that Meta's use of copyrighted books downloaded from pirate websites to train the Llama models \textit{was} fair use, but took pains to point out that the decision was based on the record before it and could well have gone the other way if the plaintiffs had developed a better factual record (\textit{Kadrey v. Meta Platforms, Inc.}, 788 F. Supp. 3d 1026, 1036–37 (N.D. Cal. 2025)).
But if the legality of use will depend on each case's particular facts and circumstances, or on the judge presiding over the case, then both \gai developers and individual creators are left without much \textit{ex ante} guidance about whether AI firms can freely take human creators' content.
So, whereas a blanket fair use privilege results in under-powered creative incentives, a highly fact-specific doctrine leads to uncertainty and potential chilling effects, particularly for risk-averse tech firms and content creators. 

What about the strict intellectual property rights model? 
Proponents of this model rely on the Coase Theorem \citep{coase1960} to argue that if property rights are clearly assigned to creators, then creators and AI firms will negotiate a price and efficient data transfers or licenses will occur. 
The standard counterargument in the literature is that {transaction costs} of different flavors---the logistical difficulty of an \gai company negotiating with millions of individual creators, overlapping or unclear ownership of IP rights, cultural and institutional barriers to mutual understanding, divergent estimates of value, and suchlike---may stand in the way of an efficient bargain \citep{ostrom1990,heller1998tragedy,heller1998}.

That is \textit{not} why the strict intellectual property rights model fails in our theoretical framework. 
We show, more pessimistically, that the system will fail even in the absence of garden-variety transaction costs---that is, even if the parties reach an agreement to sell or license human-created content to AI firms.
{ The first reason is the {market power} of AI firms.}
A monopsonist (or oligopsonist) AI firm reduces prices paid to human content creators, sacrificing the quality of content it buys in favor of higher (per-unit-of-quality) profit when it sells its \gai-generated content to consumers. 
Though this lower price is profitable for the AI firm, it reduces human incentives to create high-quality content below what is socially optimal. 
{ A second feature of this market is the {correlation} between different creators' content. While it does not directly drive this under-investment---since the buyer already accounts for the correlation when setting prices---it shapes how the under-investment is distributed. 
Our analysis reveals that the markdown falls most heavily on the creators of the most original content.} 
So the property rights model breaks down not because it fails the technological-progress objective but, counterintuitively, because it fails the creative-incentives objective, { and it does so unevenly across creators}. 
This finding inverts common wisdom about the effect of strong intellectual property rights. 

In deriving this result, we draw on recent papers in information economics, particularly those by \citet{acemoglu2022too} and \citet{bergemann2019economics}, which demonstrate in the privacy context that when individual data are correlated, selling one's data generates a negative externality by revealing information about others, leading to {over-supply} (``excessive data sharing'') relative to social optimum.
We extend this literature in several ways. 
Whereas \citet{acemoglu2022too} treat data as an {endowment}, we shift attention to the \emph{production} of creative works, where correlation implies an asymmetric substitutability.
This richer analysis allows us to investigate the differential effects of \gai on human creative works with different levels of novelty or originality, demonstrating the {uneven} under-investment (``originality penalty'') discussed above.

Further, we move beyond the static context of existing models to explore the dynamic effects of \gai interaction with potential content creators. 
The central result in our dynamic analysis is the ``curse of precision'' discussed above whereby a good model sows the seeds of its own decay.
This finding bridges two contrasting phenomena identified in machine learning: while scaling up AI models unlocks ``emergent abilities'' \citep{wei2022emergent,schaeffer2023emergent,berti2025emergent}, recursive training on self-generated data triggers ``model collapse'' \citep{shumailov2023curse,alemohammad2023self}.
We interpret and reconcile these phenomena through an economic lens: stronger AI capabilities act as a substitute for expensive human creativity, which creates a dynamic ``market for lemons'' \citep{akerlof1970market,tullis2025} that \textit{disincentivizes} the production of original content. 
Consequently, our analysis suggests that ``model collapse'' can happen---even with consistent fresh human inputs---as an economic consequence of the model's own success. 
Hence the curse of precision.

Our analysis of market failure in the presence of IP also draws on and contributes to an extensive economics literature on property rights. 
Building on the seminal contributions of \citet{coase1960} and \citet{W79}, this literature has largely focused on transaction costs, and options-to-own have been proposed as potential solutions \citep{C06, SW16}.
We are interested in the efficient allocation of property rights in data, so we tailor the model to analyze data as productive assets that can be traded in a market.
Sharing this view, a parallel line of macro-finance research treats data as a {nonrival byproduct} of economic activities, which leads to expansion and superstar firms \citep{farboodi2020long,farboodi2021model,farboodi2023data}.
By contrast, we revisit the classical insight of \citet{arrow1962economic} that valuable information comes from \emph{costly creative production}. 
However, unlike Arrow's classic setting where under-investment stems from {inappropriability} (the difficulty of excluding freeriders who copy innovators' data), we identify a unique challenge of {statistical substitutability}:\footnote{\citet{jain2010equilibrium} similarly observe substitutability among digital creative goods, but in their setting such substitutability is {semantic} and arises on the consumer side---listeners treat songs within a genre as interchangeable. Our {statistical} substitutability instead arises on the production side, driven by the correlation across creators' outputs.}
{ Even with strict intellectual property rights, the correlation across creators' outputs lets a buyer with market power mark down what it pays for content, and---as we show---it marks down the most original, least redundant content the most, discouraging precisely the creators whose contributions are hardest to replace.}

Focusing on the production of creative content on online platforms---where creators strategically compete for user attention exploiting the recommendation system---several recent works study the {optimal contest design} within an online platform \citep{jagadeesan2023competition,jagadeesan2023supply,yao2023bad,yao2023rethinking,immorlica2024clickbait,golrezaei2025contest}, a concept dating back to \citet{glazer1988optimal}.
{ Most related, a concurrent paper studies contest design when AI Overviews divert traffic from creators, proposing citation and compensation mechanisms to preserve their incentives \citep{wu2026ai}.}
While creators are strategic in both their and our models, our model focuses on {pricing} for AI training, rather than {ranking} for user consumption.
So the main friction in our model is the {imperfect substitution} of human creators, instead of the {congestion} of user attention.

Our proposed mechanism draws inspiration from the extensive economics literature analyzing the role of intermediaries in different markets \citep{RV87,FJS23} as well as intermediate legal regimes utilizing collective management organizations (CMOs) like ASCAP and BMI.
But while CMOs primarily act as centralized clearinghouses to reduce transaction costs \citep{cotter2005some,gilbert2017collective}, we argue that in the market for \gai training data, reducing transaction costs alone is not enough.
{ Because the failure stems from the buyer's market power rather than transaction costs, the solution must counter that power; and because creators' content is statistically substitutable, it must also reward originality.}
We achieve this by adopting a {two-part tariff} structure \citep{oi1971disneyland}: By combining price-per-effort payments with a fixed subsidy, we redistribute the surplus attained by socially optimal production back to all participants, thus advancing technological progress and preserving human creative incentives at the same time.

Finally, methods have been proposed in machine learning to quantify the value of training data.
For example, \citet{ghorbani2019data} and \citet{jia2019towards} proposed the {data Shapley value} to measure the marginal contribution of a data point to model performance. \citet{koh2017understanding} developed {influence functions} to approximate the counterfactual impact of the absence of a data point in low dimensional environments, while \citet{zou2025newfluence} calculated this impact in high dimensional environments. 
Various data marketplaces have also been established \citep{agarwal2019marketplace,schomm2013marketplaces,golrezaei2014pricing,muschalle2013pricing}.
But these works study fair attribution for a {fixed} dataset whereas we focus on the {production} of data.

\section{Static Model and Short-Run Market Failures}\label{sec:model}
We begin by analyzing a static market with multiple content creators and one AI firm. 
Under the classic copyright binary, either the unauthorized use of copyrighted works to train \gai models is copyright infringement---in which case the AI firm is required to obtain authorization from each copyright owner---or the use is fully shielded from copyright liability under fair use. 
We first show that, consistent with common intuition, a free-for-all system fails to provide incentives for human content creation.
We further show that the strict intellectual property rights regime leads to two market failures: first, due to the buyer's significant market power, creators invest less effort into creating high-quality content than is socially optimal (an \textit{under-investment} failure); second, this under-investment applies more forcefully to highly original or innovative content creators than to those whose content is highly correlated with others', leading to a proliferation of homogenized content (an \textit{adverse selection} which we call the ``originality penalty''). 

\subsection{Static Model Setup} \label{sec:model setup}

We formalize the interaction between human creators and an \gai developer as a market for information where production is costly and products are correlated.
Specifically, consider a market with two types of participants: 
There are $N\ge 2$ information producers (human content producers), indexed by $i=1,2,\ldots,N$, and a single buyer
(e.g., an AI firm acquiring training data from creators).
The economy centers on a fixed but unobserved state variable $X \in \mathbb{R}$, representing the best (or most useful) response to a given user query.\footnote{In some contexts, the ``best'' response has a clear, objective meaning, as when a user asks for a solution to a math problem; in other contexts, the best response is more subjective and geared to a user’s particular needs, as when a user asks a Chatbot to write a bedtime story for her child, in which case the variable $X$ measures how well the model can understand and respond to the user's preferences given her prompts.}

\paragraph{Production of Content.} We model the content created by each creator $i$ as a noisy signal $s_i$ of $X$, namely
\begin{equation}\label{eq:common factor}
s_i = X + \xi_i,~\xi_i=\varepsilon_i+\beta_i\eta,\quad \forall i=1,2,\ldots,N.
\end{equation}
As such, $\xi_i$ is a random variable denoting the deviation of an individual creator's content from $X$.
The individual deviation (or noise) $\xi_i$ has two components: an ``idiosyncratic component'' $\varepsilon_i$ whose precision is an increasing function of the individual agent's effort to produce high-quality content, and a ``common factor'' $\eta$ that is shared among all creators but weighted differently by creators (as captured by the parameter $\beta_i$).\footnote{Our additive signal structure is motivated by auction theory with interdependent values \citep{milgrom1982theory} and the value of public information in information economics \citep{morris2002social}, where agents observe noisy signals containing a statistically dependent public component. We adapt this logic to AI data markets.} 
For simplicity, we assume all $\beta_i$ are public information.

To produce content, agent $i$ chooses an effort level $h_i\ge 0$. Effort reduces the variance of the idiosyncratic noise such that $\varepsilon_i\sim \mathcal N(0,h_i^{-1})$, but incurs a strictly convex production cost $C(h_i):=\frac c2 h_i^2$ where $c > 0$ is common knowledge. We adopt the quadratic cost function---a standard formulation in contract theory \citep[see, e.g.,][]{baker1994subjective}---as a tractable abstraction of any convex cost function. We emphasize that our core findings---originality penalty, curse of precision, and the intermediated data market mechanism---extend to general convex costs. A strictly convex production cost function captures the increasing marginal cost when creators strive for higher precision via, e.g., careful research, thorough fact-checking, or creative effort.

For the common factor (the $\beta_i\eta$ term), we suggest two (mutually compatible) interpretations.  
In the first interpretation, $\eta$ represents a general {trend}, such as a stylistic or substantive convention (e.g., the seemingly star-crossed lovers improbably getting together at the end), so a large $\beta_i$ denotes agent $i$'s conformity to prevailing conventions whereas a low (near-zero) $\beta_i$ signifies agent $i$'s willingness to buck the trend and strike out on her own, to march to the beat of a different drummer. 
That is the sense in which we characterize high-$\beta$ creators as ``rank and file'' or ``run of the mill'' and low-$\beta$ creators as highly ``original'' or ``creative.'' 
And the ``originality penalty'' is the greater effort decline suffered by highly original creators, as compared to run-of-the-mill creators, when we move from the social optimum to the \gai market equilibrium. 
In the second interpretation, $\eta$ represents a common {bias}, such as socially prevalent prejudice against a group or practice or thought (e.g., pervasive prejudice against women in science or law), and $\beta_i$ measures the extent to which agent $i$ shares that bias. 
Then, the originality penalty can be interpreted as {bias amplification} by \gai---not only in the general sense that \gai-generated content is more biased (less accurate) than is socially optimal, but in the specific sense that the equilibrium effort reduction (compared to the social optimum) is {greater} for content with low common-bias than for content with high common-bias.
Depending on the context, one or the other interpretation may be more natural.  
As noted, the two interpretations are entirely compatible because, in the current model specification, the trend is biased (the addition of $\eta$ can never decrease the deviation of the creators' signals from the true state, $s_i - X$), so higher creativity implies higher quality---all else (i.e., effort) being equal.\footnote{
In principle, it's possible to conceptualize a notion of originality or creativity (or a residual of it) that is independent of accuracy or quality. 
But, without getting into the details of how such a notion could be mathematically operationalized, our conjecture is that if there is an equilibrium penalty for the kind of originality that contributes to quality (as we have shown there is), then \textit{a fortiori} there must also be an equilibrium penalty for the kind of originality that does not contribute to quality.}

A defining feature of our framework is the correlation structure of signals. 
We model the common factor $\eta$ using a hierarchical Bayesian structure:
\begin{equation}\label{eq:Bayesian mu}
\mu_\eta\sim \mathcal N(0,\gamma),\quad \eta\mid \mu_\eta \sim \mathcal N(\mu_\eta,\sigma_\eta^2),
\end{equation}
where both $\gamma\ge 0$ and $\sigma_\eta^2\ge 0$ are publicly known parameters, and $\sigma_\eta^2+\gamma>0$. The term $\mu_\eta$ represents the systematic bias of the current information environment. This hierarchical structure captures that the common bias is itself uncertain: while the buyer knows there may be a shared trend, bias, or AI-induced pattern affecting creators, it does not know its direction or magnitude \textit{ex ante}, which is captured by the prior variance $\gamma$.
This Bayesian setup is crucial for our dynamic model incorporating AI-assisted content creation in section \Cref{sec:dynamic}: 
As AI-generated content becomes pervasive, human-generated content becomes more homogeneous, the systematic bias $\mu_\eta$ becomes harder to filter out, and hence the prior variance $\gamma$ effectively increases (see \Cref{sec:dynamic} for more details).

Aside from the Bayesian perspective, \Cref{eq:common factor,eq:Bayesian mu} can be viewed through a minimax robust optimization lens: the platform estimates the ground-truth $X$ knowing only that the magnitude of systematic bias is bounded ($\mu_\eta^2\le \gamma$). This establishes an alternative interpretation for our dynamic model: when AI tools are ubiquitous, imitators' 
reliance on AI increases, thereby inducing a larger bound $\gamma$ for systematic bias; see \Cref{lem:minimax} for formal equivalence.

\paragraph{Acquisition of Data.}
To incentivize the production of high-quality  data while balancing the total purchase cost, the buyer offers a  price $p_i\ge 0$ per unit of effort to each creator $i$.\footnote{A critical premise of our Stackelberg game is that the buyer can contract on the effort vector. In reality, raw human ``effort'' (e.g., actual labor hours spent) is typically unobservable. However, modern AI platforms routinely contract on verifiable quality proxies: rigorous unit tests for coding tasks, detailed rubric scores in Reinforcement Learning from Human Feedback, or verified expert credentials. In consistency with various reduced-form approaches in contract theory, we treat $h_i$ not as unobservable labor, but as the verifiable idiosyncratic signal precision, which is directly mapped from those observable quality metrics.}
The buyer---as a leader in our Stackelberg game---has full power to decide all prices $p_i$.
The creators react to the price vector $\bm p$ by deciding an individual effort level $h_i$ that maximizes their own quasi-linear payoff:
\begin{equation}\label{eq:creator payoff}
U_i(h_i;p_i)=h_ip_i-C(h_i),\quad \forall i=1,2,\ldots,N.
\end{equation}

The first-order condition gives creators' {best response} to any price vector $\bm p$, namely $\bm h(\bm p)$, a
\begin{equation*}
h_i(\bm p)=\frac{p_i}{c},\quad \forall i=1,2,\ldots,N.
\end{equation*}
Given this one-to-one correspondence between $\bm p$ and $\bm h(\bm p)$, it is convenient to likewise define $\bm p(\bm h)$ as the price vector $\bm p$ inducing a best response of $\bm h$. Specifically, we define $p_i(\bm h)=h_ic$ for all $i$.

\paragraph{Aggregation of Data.}
After the signals $\bm s=(s_1,s_2,\ldots,s_N)'$ are produced according to \Cref{eq:common factor,eq:Bayesian mu}, the AI firm buys all these signals (by paying the promised $p_ih_i$ to each creator $i$) to construct an estimate of $X$, which we call $\hat X$. 
The firm derives revenue from the effective precision of this estimator, reflecting the fact that it can earn more when its model's responses to user queries are more useful. 
Specifically, we assume the buyer seeks to minimize the Mean Squared Error (MSE) of $\hat X$ subject to unbiasedness:
\begin{equation}\label{eq:MSE}
\min\limits_{\hat X} \quad \E[(\hat X-X)^2] \qquad ~~\text{s.t.} \quad \E[\hat X]=X.
\end{equation}
Under our model, the optimal solution to \Cref{eq:MSE} is the best linear unbiased estimator (BLUE) of $X$. While modern LLM optimizes complex objectives like cross-entropy, the minimization in \Cref{eq:MSE} is a canonical approximation for the value of information: under Gaussian assumptions, minimizing MSE is equivalent to maximizing log-likelihood, aligning with the probabilistic foundations of generative models.

The firm's revenue depends on the {effective precision} (also known as Fisher information in statistics) of its estimator $\hat X$, denoted $K$ and defined as $K(\bm h)  = \Var(\hat X)^{-1}$, with the interpretation that (a monotonic transformation of) the effective precision functions as the ``total factor productivity'' of the \gai model.
Practically speaking, a higher precision (lower variance) means that the content generated by the \gai model in response to a user query is more useful to the user, which helps the firm earn greater revenue.
The \gai developer's payoff (profit) function is the revenue earned from precision minus the price paid to creators for their content, that is,%
\footnote{In our model, the buyer's profit is \emph{linear in precision increase} rather than in variance reduction (another common option in information economics used by, e.g., \citealt{acemoglu2022too}). The reason behind our modeling choice is two-fold: First, standard information economics establishes that when information guides {downstream decisions}, utility is linear in precision (and thus {strictly convex}, instead of linear, in variance reduction) \citep[Chapter 8]{veldkamp2011information}.
Second, this convexity captures the ``emergent abilities'' of Generative AI \citep{wei2022emergent}: reducing error rates below critical thresholds triggers {phase transitions} that unlock new capabilities---beyond incremental improvements of existing ones---thereby yielding utility gains strictly exceeding the linear function of variance reduction.}
\begin{equation*}
\Pi(\bm p):=K(\bm h(\bm p))-\sum_{i=1}^N p_i h_i(\bm p),
\end{equation*}
where $\bm h(\bm p)=\bm p/c$ is the creators' best response to $\bm p$. Given the one-to-one correspondence between $\bm h$ and $\bm p$, it would be more convenient to write $\Pi$ in terms of creators' effort $\bm h$, namely
\begin{equation}\label{eq:profit}
\Pi(\bm h)=K(\bm h)-\sum_{i=1}^N p_i(\bm h) h_i = K(\bm h)-\sum_{i=1}^N c h_i^2.
\end{equation}

We summarize the sequence of play in the following Definition. 

\begin{definition}[Sequence of Play]\label{def:seq of play}
The sequence of play is as follows:
\begin{enumerate}
\item The buyer offers a unit price $\bm p=(p_1,p_2,\ldots,p_N)'$ to all creators simultaneously.
\item Creators simultaneously choose their best response $\bm h(\bm p)$, which maximizes payoff $U_i(h_i;p_i)$.
\item From the signals $\bm s=(s_1,s_2,\ldots,s_N)'$ generated by \Cref{eq:common factor,eq:Bayesian mu}, the buyer finds the estimator $\hat X$ minimizing the objective function in \Cref{eq:MSE} and receives profit $\Pi(\bm h)$.
\end{enumerate}
\end{definition}

In this framework, we contrast the {equilibrium} outcome of the game defined above with the {first best} solution.
The equilibrium is the outcome of the game specified by, \Cref{def:seq of play} where creators and the buyer pursue their own interests. We denote the resulting effort vector and prices as $\bm h^\ast$ and $\bm p^\ast$, respectively.
On the other hand, the first best solution $\bm h^{\text{sp}}$ is the outcome dictated by a benevolent social planner, who controls the production efforts to maximize social welfare:
\begin{equation}\label{eq:social welfare}
W(\bm h) := \Pi(\bm h)+\sum_{i=1}^N U_i(h_i;p_i(\bm h))=K(\bm h)-\sum_{i=1}^N \frac c2 h_i^2,
\end{equation}
where recall that $\bm p(\bm h)$ is the price vector inducing $\bm h$ as the best response, and it is excluded from social welfare as it is a pure transfer among the market participants.

\subsection{Failure of Free-for-All Regime}

We now analyze the market equilibrium under existing policy frameworks, beginning with the free-for-all (``all model training is fair use'') regime. {In our model, it is easy to make the following observation: \emph{A regime that permits \gai to take human-created content for free cannot sustain creative incentives.}}
The reason is straightforward: If the buyer offers no price to content providers, i.e., $\bm p=(0,0,\ldots,0)'$, then the only best response of creators is $\bm h(\bm 0)=(0,0,\ldots,0)'$. 
That is, creators never invest any effort in production.
This is in line with common intuition; indeed, it is the very freerider problem that justifies the existence of intellectual property (IP) laws.
Henceforth, we only focus on the ``strict intellectual property rights'' regime.\footnote{
Even if the AI firm can lawfully take human-created content for free \textit{ex post}, the firm might enter into a contract to purchase content for a positive price, thereby incentivizing content production \textit{ex ante}. 
Setting aside the practical difficulties (transaction costs) entailed in identifying and negotiating with potential content creators \textit{ex ante}, such an arrangement, though legally different from an IP-based regime (in that it's based on contract rather than property rights), would practically operate much like the ``property rights'' regime analyzed below.}

\subsection{Short-Run Market Failures: Under-Investment and Originality Penalty}\label{sec:short II}
The first market failure we identify is the {under-investment} in producing high-quality content. 
Because the buyer has full bargaining power to decide the market price $\bm p=(p_1,p_2,\ldots,p_N)'$ (also known as {monopsony power}),
it may offer a lower price, which incentivizes lower effort levels from the creators but elicits higher profit per effort. 
This leads to our first main result:
\begin{proposition}[Underpowered Creative Incentives]\label{obs:noInnov}
Under the market defined in \Cref{sec:model setup}, market equilibrium $\bm h^\ast$ and social planner solution $\bm h^{\text{sp}}$ are unique. Furthermore, $K(\bm h^\ast)<K(\bm h^{\text{sp}})$.\footnote{The under-production result in \Cref{obs:noInnov} only applies to the overall effective precision, i.e., $K(\bm h^\ast)$ compared to $K(\bm h^{\text{sp}})$. The stronger per-creator under-production---namely $h_i^\ast\le h_i^{\text{sp}}$ for all $i=1,2,\ldots,N$---is unfortunately incorrect. A counter-example is in \Cref{remark:individual creator over-produce} of \Cref{app:incentives}, where---intuitively---a highly redundant creator $i$ is set to invest $0$ effort by the social planner (hence $h_i^{\text{sp}}=0$); but due to the under-production of other creators in the equilibrium, the buyer finds it necessary to incentivize a strictly positive effort from creator $i$ (hence $h_i^\ast>0$).\label{footnote:individual creator over-produce}}
\end{proposition}

The proof is in \Cref{app:incentives}.
While \Cref{obs:noInnov} is consistent with the intuition that monopsony power leads to under-production, it paves the way for the second market failure that we identify---the {originality penalty}---a novel form of market failure distinctive to the context of data markets. 

To see why the originality penalty is non-conventional, recall that the common factor $\eta$ in \Cref{eq:common factor} implies that every creator $i$'s signal $s_i$ contains a potentially bias towards $\eta$ to some extent.
It thus follows that the buyer---who is only interested in recovering $X$ but not $\eta$---faces a dual objective: \textit{(i)} incentivizing production effort to reduce the idiosyncratic noise $\varepsilon_i$, while \textit{(ii)} mitigating the statistical bias arising from the common factor $\eta$.
One may consequently expect the AI firm to offer a better price to highly original content creators---those with low $\beta_i$---since their effort helps better recover $X$. 
Such a conjecture seems especially reasonable given that our model assumes homogeneous production costs, so producing ``original'' content does not cost more than ``conventional'' content. 
Thus, even allowing the AI firm to ``lowball'' human creators---using its monopsony power---by offering a price that does not incentivize sufficient investment in quality (\Cref{obs:noInnov}), one would expect the AI firm to lowball original content creators less than it lowballs run-of-the-mill creators in order to get the benefit of the former's individuated (non-redundant) signals.

But it turns out that the equilibrium outcome is precisely the opposite: In the market equilibrium, relative to the social optimum, the reduction in effort is {greater} for original creators than for rank-and-file creators. We call this the ``originality penalty,'' characterized in \Cref{obs:biasOrigin}.

\begin{proposition}[{Originality Penalty}] \label{obs:biasOrigin}
For each creator $i=1,2,\ldots,N$, let $R_i:=h_i^\ast/h_i^{\text{sp}}$ be the ratio of equilibrium effort to socially optimal effort. {Then in a non-degenerate market without near-redundant creators,\footnote{Formally, $\beta_i<1/\Lambda(\bm h^{\text{sp}})$ for all $i$ with $\Lambda(\bm h)$ defined in \Cref{eq:FOC lambda} of \Cref{app:incentives}. This implies the BLUE $\hat X$ defined in \Cref{eq:MSE} places a positive weight on every creator's signal; equivalently, at the social optimum, a more correlated creator (higher $\beta_i$) invests less effort (smaller $h_i^{\text{sp}}$). See \Cref{remark:non-redundant} in \Cref{app:incentives} for details. We further remark that, even under this condition, the per-creator under-investment anticipated in \Cref{footnote:individual creator over-produce} remains untrue.\label{footnote:non-redundant}} creators with a higher $\beta_i$ always retain a strictly higher $R_i$.} Furthermore:
\begin{itemize}
\item For a creator $i$ with $\beta_i\approx 0$ (an ``innovator'' who does not rely on the common factor), regardless of the remaining $\beta$'s, their ratio $R_i\approx \frac 12$.
\item For a creator $i$ with $\beta_i\gg 0$ (an ``imitator'' producing homogeneous content), the ratio $R_i>\frac 12$. 
\item In the limiting case of a homogeneous market saturated by rank-and-file creators ($N\to \infty$ and $\beta_i\equiv \beta>0$), the ratio converges to
\begin{equation*}
\lim_{N\to \infty} R_i=\frac{1}{\sqrt[3]{2}}\approx 0.79.
\end{equation*}
\end{itemize}
\end{proposition}

The counterintuitive result in \Cref{obs:biasOrigin} arises because of the interplay of {monopsony power} (i.e., the buyer has full market power) and {diminishing returns} (i.e., imitators' contribution to precision diminishes due to content correlation). 
{ On a high level, this is because at the social optimum the planner asks little effort from imitators---their content is redundant, so extra effort costs the same but barely raises precision---and much more from innovators. The AI firm's underpricing then falls unevenly: it buys less from the innovators to suppress the high cost---with convex costs, buying less means paying lower on every unit of effort---cutting them to about half the social optimum; on the other hand, it gains little from squeezing the already-low effort of those imitators.}

Specifically, innovators---because they have a near-zero weight on the common factor ($\beta_i\approx 0$)---are essentially facing the monopsonist AI firm alone. 
Their creative outputs are almost independent, in the sense that their marginal contribution to the precision remains a constant no matter how much effort they have already invested.
The social optimum is where the marginal cost of exerting effort equals this constant contribution.
However, as the production cost is strictly convex in effort, buying less can reduce the price for {all} content; a monopsonist therefore strategically restricts the purchase quantity in exchange for a lower unit price (and higher unit profit). 
Solving the tradeoff between quantity and unit price gives $h_i^\ast\approx \frac 12 h_i^{\text{sp}}$.

By contrast, imitators---given their high reliance on the common factor ($\beta_i\gg 0$) and thus high correlation with other creators---face strong diminishing returns.
One can better understand this phenomenon as if all creators were reporting the same news: The first piece of report is highly valuable, but subsequent reports are nearly devoid of value. Consequently, the socially optimal level of data production would then be ``one piece of information tells it all,'' so the level of effort is low even in the social optimum ($h_i^{\text{sp}}$). 
In the market equilibrium, even when the AI firm has full monopsony power, the price must still incentivize (roughly) this total production level: Missing a part of the first news report results in a value loss much higher than the cost saved. 
Formalizing these arguments, we prove in \Cref{app:incentives} that $h_i^\ast\ge \frac{1}{\sqrt[3]{2}} h_i^{\text{sp}}\approx 0.79 h_i^{\text{sp}}$ if $\beta_i$ is sufficiently large.

We remark that, however, \Cref{obs:biasOrigin} is \textit{not} saying that the AI firm pays a lower unit price for original creators' content than for run-of-the-mill creators' content. 
Rather, the meaning of the ``originality penalty'' is that when we move from the social optimum to the \gai market equilibrium, original creators face a greater drop in incentives to invest in effort compared to run-of-the-mill creators.
In other words, \gai penalizes originality more than it penalizes run-of-the-mill content by disincentivizing the original creators more.
In fact, given the redundancy (high correlation) of rank-and-file creators' output, eliciting a great deal of effort from them would not be worthwhile even in the social optimum. Consequently, relative to the low optimal level of effort, the AI firm's strategic underpricing does not result in a great deal of reduced effort and reduced precision.  
Original creators, by contrast, should optimally exert great effort, while the AI firm's underpricing greatly reduces their incentive to do so and, thus, the precision of their content.

The originality penalty identified in \Cref{obs:biasOrigin} poses a risk to human creative incentives: since original creators in reality (e.g., domain experts or premium publishers) typically have {outside options}, revenue generated at a sub-optimal scale may fail to cover their opportunity costs.
The exclusion of original creators' potential contributions from \gai training sets sows the seeds of the model's own decay.
We characterize this long-term crisis in the next section. 

\section{Long-Term Market Failure: The Curse of Precision}\label{sec:dynamic}

The foregoing analysis shows that standard IP rights cannot solve the \gai data crisis: { market power leads to under-investment, and---interacting with data correlation---further to an adverse selection (the originality penalty) that} essentially {refuses to debias} (by incentivizing the production of biased contents).
However, we have so far kept the model static---focusing on a {fixed snapshot} of the data market---while the \gai training pipeline is recursive: AI models are iteratively trained on web-scraped datasets that accumulate over time, but such online data is increasingly \emph{polluted} by content creators who utilize previous AI models.
This creates a closed-loop supply chain where the system's output becomes its own future input, transforming the static ``bias amplification'' problem into a dynamic ``model collapse.''

In this section, we incorporate dynamics into our analysis to study how this recursive structure affects the long-run market equilibrium. Specifically, assume time  is continuous, $t\in [0,+\infty)$. Consider two dynamic ``data stock'' state variables, $S_O(t)$ and $S_C(t)$, where  $S_O(t)$ and $S_C(t)$ represent the {cumulative stock} of ``original'' and ``correlated'' data available at time $t$. Each data stock evolves according to standard capital accumulation dynamics: It increases via the inflow of new data production, and decreases due to depreciation (see \Cref{eq:S dynamic} for the formal definition).
The two stocks of training data together determine the {current model precision} $K(t)$ at any time $t$.

Moreover, recall from \Cref{eq:Bayesian mu} that the common bias $\eta$ is centered around a {systematic bias} $\mu_\eta\sim \mathcal N(0,\gamma)$. The bias was fixed in the static case. In the dynamic model, such bias is no longer exogenous: as the AI model becomes more powerful, the rank-and-file creators 
rely on it more heavily. Consequently, the systematic bias $\mu_\eta$ may be more severe. 
We capture this {bias amplification} through a \emph{recursive bias $\gamma(t)$}, which we model as a linear function of model precision $K(t)$.
That is, stronger models induce more significant systematic biases.

To analyze the dynamics of this system, we consider a continuum of creators facing a uniform market price $p(t)$ that evolves over time $t$ (whose detailed formulation is deferred to \Cref{sec:dynamic model}).
The central question is whether the market {self-corrects over time}:
One might assume that as high-quality human data becomes scarce (i.e., the ratio between original and correlated data $S_O(t)/S_C(t)\to 0$), the market price $p(t)$ will rise, inducing innovative creators to return. We prove, to the contrary, that market failure persists. We identify a central mechanism driving this failure, which we term the ``curse of precision.'' 

The curse of precision consists of two components. First, there is an {upper bound on market price}. In particular, instead of a rising market price that reflects the scarcity of original data, the market price $p(t)$ is capped by the {diminishing marginal value} of correlated data. When either the stock of correlated data accumulates or model precision increases, this marginal value vanishes, so the price is ``trapped'' at a low level. Therefore, as formalized in \Cref{obs:price trap}, the market price is effectively {decoupled} from the scarcity of innovators. The second component is a {barrier to model precision}. If the low market price drives innovators out of the market, then the model---instead of sustaining on a massive amount of correlated data from the rank-and-file creators---hits a ``ceiling.'' In \Cref{obs:collapse}, we prove that even with an infinite supply of correlated data, model precision $K(t)$ still cannot surpass a constant upper bound.

To summarize, our core insight---the {curse of precision}---is that as the \gai model improves (more precisely estimating the state $X$, giving better answers to user queries), it sows the seeds of its own stagnation:
Superior model outputs incentivize rank-and-file human creators to rely more heavily on \gai outputs in producing the content that the model will train on, which lowers the price the AI firm is willing to pay content creators, which may in turn induce the exit of innovative creators, thereby depriving the buyer of the fresh human content it needs to keep up its good performance.
Thus, achieving a high level of precision at a given time is no guarantee of continued success and in fact may be an omen of impending failure.

\subsection{Dynamic Model Setup}\label{sec:dynamic model}

To capture the ecosystem's long-run evolution, we consider a population-level fluid model where all content creators form a {continuum}.

\paragraph{Creators and Market Structure.}
In light of the bifurcation result in \Cref{obs:biasOrigin}, we divide the large population of content creators into the two representative classes:
\begin{itemize}
\item \textbf{Innovators Producing Original Data (Type-$O$):} These creators share a near-zero loading on the common factor, denoted by $\beta_O\approx 0$.\footnote{One may generalize $\beta_O$ to a distribution of $\beta$s among innovators. For simplicity, we treat $\beta_O$ and $\beta_C$ as constants.} They can be interpreted as domain experts whose contents are based on their own opinion. Therefore, their data is uncorrelated with the common factor $\eta$.
\item \textbf{Rank-and-file Creators Producing Correlated Data (Type-$C$):} This type of creators share a high $\beta$, namely $\beta_C\gg 0$. They represent content creators who heavily rely on (and trust) AI tools. They are more exposed to the common factor $\eta$ due to AI tools.
\end{itemize}

The buyer sets a uniform price $p(t)$ for each time $t\ge 0$.\footnote{ While the micro-level model in \Cref{sec:model} allowed for creator-specific prices $p_i$, here we assume a single clearing price $p(t)$. We make this assumption because a creator's type is a statistical property of her content relative to the rest of the corpus, and is hard to verify at the point of acquisition (especially when web-scale data is bought through bulk licenses, platform agreements, or scraping settlements that price a source or category rather than individual creator's originality). However, this restriction is also not innocuous: Were the firm able to price by type, it could keep paying innovators their high marginal value, so the exit of innovators would not arise. Still, in the spirit of the market for lemons \citep{akerlof1970market}, we read the single price as a constraint forced by the non-verifiability of types.} Creators enter the market if the price $p(t)$ provides sufficient utility compared to their outside options. Instead of explicitly modeling the continuum of outside options, we denote the {mass of active creators} of type $j\in \{O,C\}$ attracted by price $p$ as $N_j(p)$. We assume both supplies are increasing in price ($N_j'>0$ for $j\in \{O,C\}$), and the supply of rank-and-file creators is {more elastic} to prices than that of innovators, namely
\begin{equation}\label{eq:elasticity dynamic}
\mathcal E_C(p)>\mathcal E_O(p),~\forall p>0,\quad \text{where }\mathcal E_j(p):=\frac{\mathrm d\ln N_j(p)}{\mathrm d\ln p},~\forall j\in \{O,C\}.
\end{equation}
In words, this assumption captures the asymmetric impact of technology on different types of content creators. For rank-and-file creators, AI tools greatly lower their barrier to entry, allowing the supply of correlated content to scale rapidly in response to market prices. Conversely, the production of original insights relies on domain expertise---scarce human intelligence that cannot be algorithmically amplified---making the supply of innovators less responsive to price changes.
As we prove in \Cref{lem:rho dynamic} of \Cref{app:dynamics}, \Cref{eq:elasticity dynamic} is equivalent to the {imitator-innovator ratio} $\rho(p):=\frac{N_C(p)}{N_O(p)}$ being strictly increasing ($\rho'>0$).
Finally---consistent with the static model---the effort of active creators is their best response to price $p$, namely $h(p):=p/c$.

\paragraph{Data Accumulation.}
We define $S_O(t)$ and $S_C(t)$ as the {cumulative stocks} of original and correlated data available to the model at time $t$ (dynamic counterpart to the total effort term $\sum_i h_i$ in the static model). They evolve according to capital accumulation dynamics:
\begin{equation}\label{eq:S dynamic}
\frac{\mathrm dS_j}{\mathrm dt}(t)=\underbrace{N_j\big (p(t)\big ) h\big (p(t)\big )}_{\text{Production Inflow}}-\underbrace{\delta S_j(t)}_{\text{Depreciation}},\quad \forall j\in \{O,C\},t\ge 0.
\end{equation}
The first term captures the total inflow: mass of creators $N_j$ multiplied by their production effort $h$. 
The second term is data deprecation. 
It captures the {concept drift} phenomenon in machine learning \citep{widmer1996learning}: due to the evolving nature of language and consensus, historical training data depreciates at a rate of $\delta\in (0,1)$ (which we assume is fixed and known).  Therefore, the buyer cannot rest on historical stocks; a continuous inflow of new data is necessary to maintain model performance. 
We follow \citet{farboodi2021model} in modeling data depreciation.

\paragraph{Recursive Bias and Precision.}
To capture the recursive nature of AI training and the reality that many content creators are increasingly utilizing powerful AI tools during production, we turn \Cref{eq:Bayesian mu} into a dynamic model. Let $K(t)$ denote the effective precision of the current model. Instead of treating the systematic bias $\mu_\eta$ and its uncertainty $\gamma$ as exogenous parameters, the common bias $\eta$ at time $t$, namely $\eta(t)$, is now generated as
\begin{equation}\label{eq:mu dynamic}
\mu_\eta(t)\sim \mathcal N(0,\gamma(t)),\quad \eta(t)\mid \mu_\eta(t)\sim \mathcal N(\mu_\eta(t),\sigma_\eta^2),\quad \forall t\ge 0,
\end{equation}
where the prior variance $\gamma(t)$ is linear in $K(t)$\footnote{In practice, rank-and-file creators can only access the {previous model} (whose precision is $K(t^-)$) but not the {newly generated model} (whose precision is $K(t)$). However, as time $t$ is continuous, it is unnecessary to distinguish between them.}
\begin{equation}\label{eq:gamma dynamic}
\gamma(t)=\lambda K(t),\quad \text{where }\lambda>0\text{ is publicly known}.
\end{equation}

The effective precision of the current mode $K(t)$ is in turn determined by the stocks $S_O(t)$, $S_C(t)$ and the bias $\gamma(t)$. \Cref{lem:precision common factor} in the Appendix shows that
the effective precision is given by:  
\begin{equation*}
K(\bm h) = \sum_{j=1}^N h_j - \frac{(\sum_{j=1}^N h_j \beta_j)^2}{(\sigma_\eta^2+\gamma)^{-1} + \sum_{j=1}^N h_j \beta_j^2}.
\end{equation*}
Therefore, in the dynamic case where the original data has mass $S_O(t)$ and loading $\beta_O$, while the correlated data has mass $S_C(t)$ and loading $\beta_C$, the model precision is defined via the following fixed-point problem:
\begin{align}
K(t)&=S_O(t)+S_C(t)-\frac{ \left (S_O(t)\beta_O+S_C(t)\beta_C\right )^2}{\big (\sigma_\eta^2+\gamma(t)\big )^{-1}+\big  (S_O(t)\beta_O^2+S_C(t)\beta_C^2\big )} \nonumber\\
&\approx S_O(t)+\frac{S_C(t)}{1+\big (\sigma_\eta^2 + \lambda K(t)\big )S_C(t)\beta_C^2},\quad \forall t\ge 0,\label{eq:precision dynamic}
\end{align}
where the last step uses $\beta_O\approx 0<\beta_C$ and $\gamma(t)=\lambda K(t)$ from \Cref{eq:gamma dynamic}.

\paragraph{AI firm Objective Function.}
We finally define the objective of the buyer. Given that \gai technologies are susceptible to creative destruction, we believe it is reasonable to assume that the AI firm is less patient than creators.\footnote{Alternatively, we can microfound this assumption by considering two AI firms who compete \`{a} la duopoly for consumers who have switching cost. In this environment, AI firms over-weight their near-future profit as it corresponds to higher market share and captive consumers.}
To simplify the exposition, we capture the difference in patience by modeling the AI firm as myopic. In particular, we assume that the AI firm cannot commit to long-term subsidies and focuses on instantaneous profit. As such, the AI firm sets the price $p(t)$ to myopically maximize {instantaneous profit}:
\begin{equation}\label{eq:profit dynamic}
\Pi_{\text{inst}}(t):=\underbrace{\frac{\mathrm dK}{\mathrm d t}(t)}_{\text{Precision Increase}}-\underbrace{\Big (N_O\big (p(t)\big )+N_C\big (p(t)\big )\Big ) p(t) h\big (p(t)\big )}_{\text{Purchase Expenditure}},\quad \forall t\ge 0.
\end{equation}

\subsection{Upper Bound on Market Price}

We analyze the long-run stability of this market. In a healthy market, if the stock of original data $S_O(t)$ drops, the scarcity should drive the price $p(t)$ up, thereby attracting innovators back.
However, due to the common factor $\eta$, the marginal value of correlated data {decays over time}, thus depressing the price: unlike the marginal contribution of innovators---which stays bounded away from zero ($0<\frac{\partial K}{\partial S_O}(t)\le 1$)---that of rank-and-file creators exhibits diminishing returns analogous to \Cref{obs:biasOrigin}.
Specifically, we derive an upper bound $\kappa_C(t)$ on the marginal increase in effective precision from rank-and-file data. As such, $\kappa_C(t)$ is the upper bound on the {marginal value of rank-and-file data.} 
\begin{equation}\label{eq:kappa_C dynamic}
\frac{\partial K}{\partial S_C}(t)\le \kappa_C(t):=\left (1+\big (\sigma_\eta^2+\lambda K(t)\big ) S_C(t) \beta_C^2\right )^{-2}\le 1.
\end{equation}
 \Cref{lem:precision vs data dynamic} formally derives \Cref{eq:kappa_C dynamic}. From \Cref{eq:kappa_C dynamic}, there are two forces driving down $\frac{\partial K}{\partial S_C}$. The first force is more intuitive and simply corresponds to {accumulation of more data}. As the stock of correlated data $S_C(t)$ grows, the marginal value of additional correlated data diminishes. This aligns well with standard market dynamics: scarcity commands a premium, whereas redundancy leads to price depreciation. The second force is more specific to the context of \gai model training and corresponds to the {recursive bias embedded in a more precise model}. This force is an amplification effect introduced by $K(t)$: As the model becomes more precise, the rank-and-file creators rely even more on the common bias, which in turn amplifies the bias going forward. That is, a stronger model (larger $K(t)$) effectively ``pollutes'' the rank-and-file creator data with amplified bias, which further drives down the value of rank-and-file creators' data beyond the redundancy effect.

While a standard market stabilizes under the {data accumulation} effect via price adjustment, the {recursive bias} effect makes the market unrecoverable. The AI firm---facing an overwhelming supply of rank-and-file creators whose data has vanishing marginal value---fails to raise the price $p(t)$ to attract innovators. The next proposition captures this phenomenon.

\begin{proposition}[{Upper Bound on Market Price}]\label{obs:price trap}
At any time $t\ge 0$, the market price $p(t)$ cannot exceed a threshold determined by $\kappa_C(t)$, namely $p(t)\le \bar p(\kappa_C(t))$ for some strictly increasing function $\bar p(\cdot)$.
Specifically, when $\kappa_C\to 0^+$, this upper bound $\bar p(\kappa_C)$ converges to a ``price trap'' $p^\ast$:
\begin{equation}\label{eq:price trap}
\lim_{\kappa_C\to 0^+} \bar p(\kappa_C)=p^\ast,\quad \text{where } p^\ast = \left (1+\frac{N_C(p^\ast)}{N_O(p^\ast)}\right )^{-1}.
\end{equation}
\end{proposition}

\Cref{obs:price trap} illustrates that  a smaller $\kappa_C(t)$ induces a lower market price $p(t)$. Moreover, it  characterizes the limiting behavior of the market price: it shows that as $\kappa_C(t)\to 0^+$---driven by the accumulation of correlated data or the recursive bias of a more precise model---the market price is capped at the fixed level $p^\ast$ specified in \Cref{eq:price trap}.

In particular, as the stock of correlated data accumulates (i.e., $S_C(t)\to \infty$) or the recursive bias amplifies due to powerful \gai models (i.e., $K(t)\to \infty$), the marginal value of rank-and-file creators' data in \Cref{eq:kappa_C dynamic} vanishes. Consequently, the market price $p(t)$ is asymptotically trapped below $p^\ast$, which is induced by the {imitator-innovator ratio} $\frac{N_C(p^\ast)}{N_O(p^\ast)}$.
Therefore, if widespread availability of \gai lowers the barriers to entry for rank-and-file creators, the supply of rank-and-file creators $N_C$ increases far faster than that of innovators $N_O$. This further pushes the trap level $p^\ast$ down towards zero, which creates a risk: if $p^\ast$ is insufficient to attract \textit{any} innovators (domain experts with outside options), what will happen to the \gai model?

\subsection{Barrier of Model Precision and ``Curse of Precision''}
To answer this question, we now suppose all innovators exit the market ($S_O(t)\to 0$) and see whether the buyer can survive solely on massive amounts of rank-and-file creator data ($S_C(t)\to \infty$). The widely believed machine learning ``scaling laws'' \citep{kaplan2020scaling,hoffmann2022training} suggest that the \textit{quantity} of training data can substitute for \textit{quality}: by training on a sufficiently large amount of synthetic data, the AI model might suddenly be able to discover the ground-truth.
However, we prove that the quantity of synthetic data cannot compensate for the scarcity of original data:
\begin{proposition}[Barrier of Model Precision]\label{obs:collapse}
Suppose innovators exit the market ($S_O(t) = 0$). Even if the stock of rank-and-file creator data grows indefinitely ($S_C(t) \to \infty$), the equilibrium model precision $K(t)$ converges to a finite upper bound $\bar K$ determined by the recursive bias coefficient $\lambda$ in \Cref{eq:gamma dynamic}:
\begin{equation}\label{eq:collapse}
\lim_{S_C(t) \to \infty} K(t) = \bar K := \Big (-\sigma_\eta^2 + \sqrt{\sigma_\eta^4 + 4\lambda / \beta_C^2}\Big) \Big / 2\lambda.
\end{equation}
\end{proposition}

\Cref{obs:collapse} captures a failure mode similar to the ``model collapse'' in empirical machine learning \citep{shumailov2023curse,alemohammad2023self}. This term typically refers to the performance degeneration when an AI model is trained on {self-generated data}---specifically, the model's output distribution loses variance and {collapses} onto a few modes (hence the name). Crucially however, our result reveals that this failure is not limited to such closed-loop training. We show that even if the buyer is continuously acquiring {fresh data} from human creators---via a market containing both innovators and rank-and-file creators---the model still hits a performance ceiling.

We call this the counterintuitive result the ``curse of precision'': a more powerful \gai model induces stronger correlation among rank-and-file creators, which drives down the market price (\Cref{obs:price trap}) and potentially hurts the model's long-run performance (\Cref{obs:collapse}). We then ask how to design a mechanism for the \gai data market to recover market efficiency and also participation.

\section{Mechanism Design: Restoring Efficiency and Participation}\label{sec:mechanism}

\Cref{sec:model,sec:dynamic} have revealed the inadequacies of existing approaches to the design of \gai data markets.
Intuitively, we showed that if AI firms can simply take human content for free, then humans are deprived of incentives to invest in creating high-quality content. 
More interestingly, we demonstrated that IP rights are insufficient to incentivize the creation of high-quality content. 
Although IP rights are meant to incentivize creativity, in fact they under-power creative incentives, especially for innovative creators. 
This market failures create a trajectory toward ``model collapse'' where the \gai model eventually starves itself of the novel human input it needs to remain robust.

In this section, we move from diagnosis to treatment. 
We propose and analyze the architecture of a {data intermediary} that uses \text{two-part pricing contracts} to restore human creative incentives while simultaneously providing \gai models with the inputs they need to improve. 

The intermediary in our solution is similar to collective management organizations (CMOs) that represent a group of content owners, such as performing rights organizations like ASCAP and BMI or reproduction rights organizations like the Copyright Clearance Center (CCC). 
However, unlike traditional copyright CMOs, whose primary function is to reduce transaction costs \citep{cotter2005some,gilbert2017collective}, { the role of our intermediary is dictated by a different problem. 
The failures we identify---under-investment and the originality penalty---arise from the AI firm's market power, and survive even when bargaining is frictionless (\Cref{obs:noInnov}); reducing transaction costs alone would therefore not restore socially optimal content creation. 
By aggregating the rights of a critical mass of creators and negotiating as a single counterparty, the intermediary counters the buyer's market power and secures creators a share of the surplus their content generates.}\footnote{If the intermediary is a private party, there can be principal-agent problems between the content creators and the intermediary. Currently, we focus on the design of a regulated intermediary and abstract from 
delegation frictions.}

In addition to aligning individual incentives with social welfare, our intermediary ensures the active participation of all parties. 
Specifically, the intermediary's economic function is twofold. 
{ First, it determines how the surplus is divided among creators. In an unintermediated market, the buyer's markdown falls most heavily on the most original creators---the originality penalty of \Cref{obs:biasOrigin}. Managing the creators as a portfolio, the intermediary instead apportions their collective compensation by each creator's contribution to the joint value of the data---using weights inspired by the Aumann--Shapley value---so that the most original creators, who contribute the most, receive correspondingly larger shares rather than the deepest cut.}

Second, the intermediary {ensures participation via adequate transfers}. 
To resolve the tension between efficient production (which requires unit prices equal to marginal production costs) and incentive preservation (which requires subsidies to meet outside options), we propose a nonlinear pricing schedule. Specifically, we implement a {two-part tariff}: a \textit{unit price} set to equate the buyer's demand with social marginal costs, thus inducing socially optimal effort, and a fixed \textit{lump-sum transfer} (licensing fees or subsidies) to redistribute the surplus back to all agents, ensuring that their participation constraints are met without distorting production incentives.

In what follows, we formally introduce our intermediated market design.
We will assume, for purposes of this analysis, that the intermediary is fully informed about the creators. 
Though this assumption is unlikely to be satisfied in the real world, it is a useful benchmark in that any inefficiencies identified here will persist in a more complex model with information asymmetries. 
Note, moreover, that the market failures we identified above persist {even in the full-information benchmark}, so the proposed solution is calibrated to the problems. 

\subsection{Institutional Design:  Intermediated Data Market}

Consider $N \geq 2$ human content creators, a single AI firm, and an intermediary with a fixed operational cost $F \ge 0$. 
The intermediary {negotiates with both sides}:
\textit{(i)} To the creators, the intermediary promises an individual {payment} $P_i(\bm h)$ conditional on the $N$  creators jointly invest efforts $\bm h=(h_1,h_2,\ldots,h_N)'.$
\textit{(ii)} To the AI firm, the intermediary offers the $N$-dimensional signal generated by all the data creators investing efforts $\bm h=(h_1,h_2,\ldots,h_N)'$, in return for a {transfer} $T(\bm h)$ to receive all signals $\bm s$.
That is, the intermediary determines two functions $\bm P\colon \mathbb R_{\ge 0}^N \to \mathbb R^N$ and $T\colon \mathbb R_{\ge 0}^N \to \mathbb R$. We call tuple $\langle \bm P,T\rangle$ a \textit{mechanism}.
As in \Cref{sec:model}, we assume the effort levels $\bm h$ are observable and contractible.

Analogous to \Cref{eq:creator payoff}, the payoff of a creator $i$ under the mechanism $\langle \bm P,T\rangle$ is specified as
\begin{equation*}
U_i(\bm h;\bm P):=P_i(\bm h)-C(h_i),\quad \forall i=1,2,\ldots,N,
\end{equation*}
where $C(h_i)=\frac c2 h_i^2$ is the production cost.
Meanwhile, the profit of the AI firm  is given as
\begin{equation*}
\Pi(\bm h;T):=K(\bm h)-T(\bm h),
\end{equation*}
where recall that $K(\bm h)$ is the effective precision induced by effort vector $\bm h$.

To capture the participation constraints of each agent, we denote the AI firm's outside option profit by $\underline \Pi\ge 0$ and each creator $i$'s outside option payoff by $\underline u_i\ge 0$. We assume all outside options are publicly known.
The intermediary seeks a mechanism $\langle \bm P,T\rangle$ satisfying:
\begin{enumerate}
\item \textbf{Incentive Compatibility (IC):} All creators invest the socially optimal efforts $\bm h^{\text{sp}}$:
\begin{align}\label{eq:IC}
\displaystyle h_i^{\text{sp}}\in \argmax_{h_i\ge 0} U_i\big ((h_i,\bm h_{-i}^{\text{sp}});\bm P\big ),\quad \forall i=1,2,\ldots,N.\tag{IC}
\end{align}
\item \textbf{Budget Balance (BB):} The intermediary breaks even (neither deficit nor surplus):\footnote{Note that our Budget Balance \eqref{eq:BB} differs from the global budget balance requirement analyzed by \citet{holmstrom1982moral}, which is known to be impossible when also requiring Incentive Compatibility \eqref{eq:IC}. Here, creators' total income $T(\bm h)$ need not equal the social value $K(\bm h)$; the difference $\Pi=K(\bm h)-T(\bm h)$ is absorbed by the AI firm, who acts as the ``budget breaker'' to extract the surplus or bear the deficit.}
\begin{align}\label{eq:BB} 
\displaystyle T(\bm h)=F+\sum_{i=1}^N P_i(\bm h),\quad \forall \bm h\in \mathbb R_{\ge 0}^N.\tag{BB}
\end{align}
\item \textbf{Individual Rationality (IR):} At the equilibrium $\bm h^{\text{sp}}$, no agent (human creator or AI firm) is worse off than their outside option:
\begin{align}\label{eq:IR}
U_i(\bm h^{\text{sp}};\bm P)\ge \underline u_i,\quad \forall i=1,2,\ldots,N;\qquad \text{and} \qquad \displaystyle \Pi(\bm h^{\text{sp}};T)\ge \underline \Pi.\tag{IR}
\end{align}
\end{enumerate}

We model the negotiation between the intermediary---on behalf of all creators---and the buyer using Nash bargaining. Let $\alpha\in [0,1]$ denote the intermediary's bargaining power (and $(1-\alpha)$ that of the buyer), which is an exogenous parameter measuring the collective bargaining power of all creators. The {Nash bargaining solution} $\langle \bm P,T\rangle$ is selected to maximize the Nash product:
\begin{equation}\label{eq:Nash bargaining}
\max \underbrace{\Bigg (\sum_{i=1}^N \left (U_i(\bm h^{\text{sp}};\bm P)-\underline u_i\right )\Bigg )}_{\text{Creators' Net Surplus}}{\vphantom{\Bigg )}}^{\alpha} \underbrace{\Bigg (\Pi(\bm h^{\text{sp}};T)-\underline \Pi\Bigg )}_{\text{Buyer's Surplus}}{\vphantom{\Bigg )}}^{1-\alpha}\quad \text{subject to \eqref{eq:IC}, \eqref{eq:IR}, and \eqref{eq:BB}}.
\end{equation}

{ We summarize the resulting sequence of play in the following definition, in parallel with \Cref{def:seq of play}.}

\begin{definition}[Sequence of Play, Intermediated Market]\label{def:seq of play intermediated}

The sequence of play is as follows:
\begin{enumerate}
\item The intermediary and the buyer bargain over the mechanism $\langle \bm P,T\rangle$ according to the Nash bargaining procedure in \Cref{eq:Nash bargaining}, where the intermediary has a bargaining power of $\alpha$. The intermediary then commits to the payment rules $\bm P$ and the transfer rule $T$.
\item Creators simultaneously choose efforts to maximize payoff $U_i(\bm h;\bm P)$. From Equation \eqref{eq:IC}, their best response must equal the social optimum $\bm h^{\text{sp}}$.
\item From the signals $\bm s$ generated by $\bm h^{\text{sp}}$, the buyer receives profit $\Pi(\bm h^{\text{sp}};T)$ and pays the transfer $T(\bm h^{\text{sp}})$, while the intermediary disburses payments $P_i(\bm h^{\text{sp}})$ and breaks even by Equation \eqref{eq:BB}.
\end{enumerate}
\end{definition}

\subsection{Linear Pricing Fails on Participation}

We start with a simple class of mechanisms---linear pricing---and show why it falls short. 
Specifically, consider the following family of mechanisms:
\begin{equation}\label{eq:linear pricing}
P_i(\bm h)=p_ih_i,\quad \forall i=1,2,\ldots,N,
\end{equation}
where $p_1,p_2,\ldots,p_N\ge 0$ are the unit prices for each creator. From \Cref{eq:BB}, the corresponding transfer function $T(\bm h)$ must be $T(\bm h)=\sum_{i=1}^N p_i h_i+F$. We begin by analyzing \Cref{eq:IC}:
\begin{lemma}[{Marginal Prices and IC}]\label{lem:IC}
Under linear pricing, the creators' best response coincides with the social optimum $\bm h^{\text{sp}}$ if and only if the unit prices are set as $p_i = \frac{\partial K}{\partial h_i}(\bm h^{\text{sp}})=ch_i^{\text{sp}}$ for all $i$.
\end{lemma}

Since the unit prices $\bm p$ are uniquely pinned down by \Cref{lem:IC}, the linear pricing mechanism loses the flexibility to redistribute surplus. Therefore, it fails to satisfy the participation constraints in \Cref{eq:IR}. Specifically, as shown in \Cref{obs:linear pricing}, the social surplus is almost fully used to incentivize creators' efficient production, leaving the buyer in deficit. So the intermediary, while protecting human creators' incentives, drives the \gai developer out of the market.
But there is nothing the intermediary can do if we insist on a linear pricing mechanism, as such $\langle \bm P,T\rangle$ is the unique one ensuring \Cref{eq:IC,eq:BB}. 
This {inflexibility} explains why markets failed in \Cref{sec:model,sec:dynamic}: 
Under linear regimes, it is impossible to balance all parties' incentives.

\begin{lemma}[{Linear Pricing Fails on Participation}]\label{obs:linear pricing}
There exists a unique linear pricing mechanism $\langle \bm P,T\rangle$ satisfying \Cref{eq:BB,eq:IC}. Under this mechanism, each creator $i$'s payoff is $U_i(\bm h^{\text{sp}};\bm P)=\frac c2 (h_i^{\text{sp}})^2$, while the buyer's profit is
\begin{equation*}
\Pi(\bm h^{\text{sp}};T)=K(\bm h^{\text{sp}})-\big (\nabla K(\bm h^{\text{sp}})\big )'\bm h^{\text{sp}}-F.
\end{equation*}
In the special case where $\beta_i\equiv 0$, we have $\Pi(\bm h^{\text{sp}};T)=-F$, thus violating \Cref{eq:IR} whenever $F+\underline \Pi>0$.
\end{lemma}
All formal proofs omitted in this section are in \Cref{sec:mechanism appendix}.

\subsection{Optimal Contract: Two-Part Pricing}\label{subsec:two part}
Since linear pricing is pinned down by incentive constraints (\Cref{lem:IC}) and fails to cover fixed costs, we must introduce an additional component to ensure Individual Rationality \eqref{eq:IR}.
It turns out to be sufficient to limit attention to {two-part tariff} mechanisms \citep{oi1971disneyland}, where payment and transfer functions take the following affine forms:\footnote{A natural question is whether equipping the AI firm itself with the two-part tariff of \Cref{eq:two-part tariff} would already restore efficiency. The answer is formally affirmative but economically vacuous: The firm would elicit $\bm h^{\text{sp}}$ via the marginal price, and consequently extract \emph{all} social surplus through a negative $B_i=\underline u_i-\frac c2 (h_i^{\text{sp}})^2<0$. This forces creators to wire substantial upfront fees \emph{before} earning any production wages. On the creator side, this violates standard expectations that the \emph{buyer} of the product (especially product protected by property rights)---instead of the seller---is the one who pays. Such a reverse payment would make many sellers suspicious and unwilling to enter into the arrangement. On the firm side, this requires a credible commitment---delivering the promised wages after the fee is collected---which no profit-maximizing buyer would have an incentive to honor. { Such a scheme also leaves creators exactly at their outside options, with no share of the surplus they generate; securing them such a share is one role of the intermediary.} Finally, the affirmative answer also relies heavily on the full information on creators' $\bm \beta$ and $\bm h$, and introducing private information and/or moral hazard to the interaction between the AI firm and creators will break this result. We leave all these issues for the next version of this work.}
\begin{equation}\label{eq:two-part tariff}
P_i(\bm h)=ch_i^{\text{sp}}\times h_i+B_i,~\forall i=1,2,\ldots,N;\quad T(\bm h)=\sum_{i=1}^N ch_i^{\text{sp}}\times h_i+L.
\end{equation}
In \Cref{eq:two-part tariff}, $B_i$ is a {base subsidy} (or fee) for creator $i$, and $L$ is a {lump-sum licensing fee} paid by the buyer.
We define the {net social surplus} $\mathcal S(\bm h)$ relative to outside options as:
\begin{equation}\label{eq:surplus two part}
\mathcal S(\bm h):=\underbrace{K(\bm h)\vphantom{\Bigg ]}}_{\text{Social Value}}-\underbrace{\left [F+\sum_{i=1}^N C(h_i)\right ]}_{\text{Total Cost}}-\underbrace{\left [\underline \Pi+\sum_{i=1}^N \underline u_i\right ]}_{\text{Outside Options}},\quad \forall \bm h\in \mathbb R_{\ge 0}^N.
\end{equation}

\begin{algorithm}[t!]
\floatname{algorithm}{Mechanism}
\caption{Two-Part Tariff Mechanism}
\label{alg:mechanism}
\begin{algorithmic}[1]
\small
\Require Number of creators $N\ge 2$, production cost $c>0$, parameters $\bm \beta,\gamma,\sigma_\eta^2$ in \Cref{eq:common factor,eq:Bayesian mu}, buyer's outside option $\underline \Pi$, creators' outside options $(\underline u_i)_{i=1}^N$, and Nash bargaining power $\alpha\in [0,1]$.
\State The intermediary calculates the social optimum $\bm h^{\text{sp}}$ maximizing social welfare $W$ in \Cref{eq:social welfare}.
\State The intermediary calculates the social surplus at $\bm h^{\text{sp}}$ defined in \Cref{eq:surplus two part}, namely $\mathcal S(\bm h^{\text{sp}})$.
\State The intermediary decides the {lump-sum transfer} $L$ from the buyer as:
\algeq[eq:transfer two part]{$\displaystyle L=K(\bm h^{\text{sp}})-\sum_{i=1}^N C'(h_i^{\text{sp}})\times h_i^{\text{sp}}-(1-\alpha)\mathcal S(\bm h^{\text{sp}})-\underline \Pi.$}
\vspace{7pt}
\State The intermediary decides the {base subsidy} $B_i$ for each creator $i=1,2,\ldots,N$ as:
\algeq[eq:subsidy two part]{$\displaystyle B_i=\phi_i\alpha \mathcal S(\bm h^{\text{sp}})+\underline u_i- C(h_i^{\text{sp}}),$}
where $\bm \phi=(\phi_1,\phi_2,\ldots,\phi_N)'$ is a probability distribution (i.e., $\phi_i\ge 0$ and $\sum_{i=1}^N \phi_i=1$) defined as
\algeq[eq:AS weights]{$\displaystyle \phi_i\propto h_i^{\text{sp}}\int_0^1 \frac{\partial \mathcal S}{\partial h_i}(\tau \bm h^{\text{sp}})\,\mathrm d\tau\qquad 
\text{s.t.}\quad \sum_{i=1}^N \phi_i=1.$}
\vspace{7pt}
\State The intermediary commits to the {two-part tariff} mechanism $\langle \bm P,T\rangle$ specified by \Cref{eq:two-part tariff}.
\end{algorithmic}
\end{algorithm}

At the social optimum $\bm h^{\text{sp}}$, the surplus $\mathcal S(\bm h^{\text{sp}})$ is divided via Nash bargaining. 
The intermediary secures a share $\alpha\in [0,1]$ for all creators, while the remaining $(1-\alpha)\mathcal S(\bm h^{\text{sp}})$ goes to the buyer.
Thus, the transfer from the buyer to the intermediary, namely $T(\bm h^{\text{sp}})$, must satisfy
\begin{equation*}
K(\bm h^{\text{sp}})-T(\bm h^{\text{sp}})=\Pi(\bm h^{\text{sp}};T)=(1-\alpha) \mathcal S(\bm h^{\text{sp}})+\underline \Pi.
\end{equation*}
This uniquely determines the lump-sum licensing fee $L$ from the buyer as \Cref{eq:transfer two part}.

To distribute the creators' share $\alpha \mathcal S(\bm h^{\text{sp}})$ fairly, we define weights $\bm \phi=(\phi_1,\ldots,\phi_N)'$ inspired by the Aumann--Shapley value \citep{aumann1974values}. Specifically, in \Cref{eq:AS weights}, each $\phi_i$ is proportional to creator $i$'s cumulative contribution to the surplus along the diagonal path from $\bm 0$ to $\bm h^{\text{sp}}$.
We set each creator $i$'s utility to their fair share $\phi_i \alpha \mathcal S(\bm h^{\text{sp}})$ plus their outside option $\underline u_i$:
\begin{equation*}
ch_i^{\text{sp}}\times h_i^{\text{sp}}+B_i-C(h_i^{\text{sp}}) = U_i(\bm h^{\text{sp}};\bm P)=\phi_i\alpha \mathcal S(\bm h^{\text{sp}})+\underline u_i.
\end{equation*}
This yields the required base subsidy $B_i$ for each creator $i=1,2,\ldots,N$ as \Cref{eq:subsidy two part}.

We therefore obtain a two-part tariff mechanism $\langle \bm P,T\rangle$ satisfying \Cref{eq:IC,eq:IR,eq:BB} at the same time. Moreover, due to the Nash bargaining procedure, it also maximizes the Nash product defined in \Cref{eq:Nash bargaining} among \textit{all feasible mechanisms}, including those non-affine ones. We summarize our findings regarding \Cref{alg:mechanism} in \Cref{obs:two part}:
\begin{proposition}[Two-Part Tariff Attains Efficiency and Participation]\label{obs:two part}
If $\mathcal S(\bm h^{\text{sp}})>0$, i.e., the market is efficient, the two-part tariff mechanism $\langle \bm P,T\rangle$ defined in \Cref{alg:mechanism} satisfies:
\begin{enumerate}
\item \textbf{Incentive Compatibility \eqref{eq:IC}}: creators' best response is $\bm h^{\text{sp}}$.
\item \textbf{Global Budget Balance \eqref{eq:BB}:} the intermediary strictly breaks even.
\item \textbf{Individual Rationality \eqref{eq:IR}:} all participants earn at least their outside options.
\end{enumerate}
Furthermore, this $\langle \bm P,T\rangle$ maximizes the Nash product in \Cref{eq:Nash bargaining} among \emph{all} feasible mechanisms satisfying \Cref{eq:IC,eq:IR,eq:BB}.
\end{proposition}

\subsection{Remaining Challenges and Legal Implementation}
In this section, we lay out some remaining questions and challenges to be addressed as we continue to refine our market design proposal.
As discussed above, the need for our solution emerges out of the failure of the two options usually available to courts.
Neither a system of blanket immunity via fair use nor a legal regime of full control over access to content by copyright owners will be sufficient to preserve creative incentives, so a workable solution must transcend the traditional copyright binary. 
Although our diagnosis of market failures in this context is novel, our proposal for going beyond the common copyright binary is by no means without precedent. 
The US Copyright Act (to say nothing of other countries' laws) contains a number of such ``intermediate'' regimes.\footnote{Examples include: 17 U.S.C. \S\,111 (governing secondary transmissions of broadcast programming by cable); 17 U.S.C. \S\,112 (governing ephemeral recordings); 17 U.S.C. \S\,114 (governing digital audio transmission of sound recordings); 17 U.S.C. \S\,115 (governing compulsory licenses to make and distribute phonorecords of nondramatic musical works); 17 U.S.C. \S\,118 (governing noncommercial broadcasts); 17 U.S.C. \S\,119 (governing secondary transmissions of distant TV station signals by satellite carriers); 17 U.S.C. \S\,122 (governing secondary transmissions of local TV station signals by satellite carriers).} 
We plan to draw lessons from these regimes by understanding their precise legal workings, the political and economic forces that shaped them, and their policy successes and failures.
Although this research is only just beginning, it's useful preliminarily to point out some comparisons between these precedents and the \gai context we target. 

Broadly speaking, existing intermediate regimes share two features with the present \gai context. 
First, like the intermediate regime we propose, existing intermediate regimes were often responses to a disruptive technology that upended settled ways of doing things \citep{shahshahani2018}.
Second, consistent with our intuitions and preliminary results about the role of data intermediaries, these legislative compromises often involve organizations to coordinate collective action among copyright holders (or, occasionally, among other stakeholders). 
These organizations include the Mechanical Licensing Collective, SoundExchange, the Copyright Clearance Center, AP Images, Artists Rights Society, reprographic rights organizations, ASIP, and performance rights organizations such as ASCAP, BMI, SESAC, and Global Music Rights.
Some of these organizations are created or sanctioned by legislation (or by regulations enacted pursuant to legislation), and some of them are  outgrowths of ``organic'' coordination among copyright holders or other stakeholders. 

Notwithstanding these similarities, existing intermediate regimes differ from our proposed regime in one crucial respect: 
Under existing regimes, the primary function of intermediaries is often to reduce transaction costs and smoothen frictions in collective action---exactly the sort of problems which one would intuit to be at the root of bargaining breakdown, but which were \textit{not} the main mechanism driving the failure of the property-rights model in our theoretical framework (see \Cref{obs:noInnov} and surrounding discussion). 
Given that, in our model, market failures happen even in the absence of garden-variety transaction costs, the function of intermediaries in our market design must { do more than reduce transaction costs}.
Hence the foregoing discussion of mechanisms to { counter the buyer's market power} and subsidize originality.

Beyond highlighting these similarities and differences, researching existing intermediate regimes points up two important questions of institutional design.
The first question relates to the {compulsory or voluntary} character of participation in an intermediary organization:
Once the intermediary is constructed, can its benefits be realized by content creators' uncoordinated, self-interested equilibrium behavior or do collective-action problems require centralized enforcement and a legal obligation to join?  
This question implicates issues of collective action that have been of perennial interest in political economy at least since \cite{olson1965}.\footnote{
From a legal perspective, it is important to note that the facilitation of collective action by means of intermediaries probably requires a statutory exemption from antitrust liability, and the aforementioned intermediated regimes under the Copyright Act generally contain such exemptions.}

Another important question is that of {price setting}. 
It is widely understood that, in many contexts, a centralized system is not as good at determining the value of assets as the decentralized information-aggregating mechanism of the market \citep{hayek1945}. 
The way an intermediate regime, such as a compulsory-licensing scheme, sets the price at which an exchange must take place is therefore an important ingredient of the regime's success or failure.
Pricing content may be easy if that content has a market outside the use for which the compulsory license is sought, in which case one can piggyback on the market price in setting the license fee. 
But if the primary market for the content is the use for which the compulsory license is sought, then there is no ready recourse to the market price signal. 
In the context of a data market for AI, pricing may not be difficult in the beginning because the data in question have had a long history of consumer demand predating the advent of AI, which history can be relied upon for pricing; however, going forward, if we anticipate the primary (direct) market for content to be as fodder for \gai models, then pricing becomes more difficult. 
The question has at least two aspects: (1) determining the price at which the intermediary sells the content to the AI firm, (2) determining how to apportion the price received by the intermediary among different content creators. 
Our analysis assumes that the originality and long-term value of content is \textit{known} to the intermediary, but the elicitation of such information (assuming it is even {privately} available) is a central difficulty of centralized mechanisms with which our solution must contend. 

\section{Conclusion}

This paper studies the economics of data markets for \gai. { The force that depresses content creation is the market power of AI firms; the statistical substitutability of human content---the correlation across creators' outputs---does not create this shortfall but determines who bears it.} In static markets, we show that the interaction between correlation and AI firms' market power generates an originality penalty that disproportionately harms creators producing uncorrelated and socially valuable content, thereby weakening incentives precisely for the creators with the highest marginal contribution. Extending the analysis to a dynamic setting with AI-assisted creation, we identify a curse of precision: improvements in \gai models induce increasingly homogenized data production, creating an upper bound on market prices and potentially causing innovative creators to exit. As a result, even a continued inflow of fresh human data may fail to prevent long-run model collapse absent an appropriate market design.

To address these failures, we propose a market design with two components, institutional design and contract design. In particular, we propose a data intermediary that uses a two-part tariff contract that implements the social optimum while satisfying participation constraints and maintaining budget balance. The mechanism restores long-run participation for both model performance and human creativity by aligning incentives for innovation with technological progress. We conclude by highlighting two remaining challenges for implementation: whether participation in intermediary organizations can emerge through decentralized equilibrium behavior or instead requires centralized enforcement and statutory coordination, and how intermediaries can estimate creators' marginal contributions in environments where outside markets for content are thin or nonexistent. Integrating the proposed mechanism with data valuation techniques therefore remains an important direction for future research.


\bibliographystyle{plainnat}
\bibliography{references}

\onecolumn
\appendix
\renewcommand{\appendixpagename}{\centering \LARGE Technical Appendices}
\appendixpage

\startcontents[section]
\printcontents[section]{l}{1}{\setcounter{tocdepth}{2}}

\section{Proofs for \Cref{sec:model}}\label{app:incentives}

In this section, we first solve the market equilibrium resulting from \Cref{def:seq of play} to derive the first-order conditions specifying $\bm h^\ast$. Likewise, we also solve the first-order conditions for the first best solution $\bm h^{\text{sp}}$. We then prove \Cref{obs:noInnov,obs:biasOrigin} based on these conditions. Finally, in \Cref{lem:minimax}, we present an alternative robust minimax optimization view of our Bayesian model \Cref{eq:Bayesian mu}.

\subsection{First-Order Conditions}
\begin{lemma}[First-Order Conditions for $\bm h^\ast$ and $\bm h^{\text{sp}}$]\label{lem:FOC common factor}
Define
\begin{equation}\label{eq:FOC lambda}
\Lambda(\bm h):=\frac{\sum_{j=1}^N h_j \beta_j}{(\sigma_\eta^2+\gamma)^{-1} + \sum_{j=1}^N h_j \beta_j^2}.
\end{equation}
Then the first-order conditions of the buyer and social planner give
\begin{equation}\label{eq:effort using Lambda}
h_i^\ast=\frac{1}{2c}\big (1-\beta_i \Lambda(\bm h^\ast)\big )^2,\quad h_i^{\text{sp}}=\frac{1}{c}\big (1-\beta_i \Lambda(\bm h^{\text{sp}})\big )^2,\quad \forall i=1,2,\ldots,N.
\end{equation}
\end{lemma}
\begin{proof}
For the buyer who maximizes profit $\Pi$ (recall \Cref{eq:profit}), namely
\begin{equation*}
\Pi(\bm h)=K(\bm h)-\sum_{i=1}^N h_i p_i=K\left (\bm h\right )-\sum_{i=1}^N \frac{p_i^2}{c},
\end{equation*}
the first-order condition gives
\begin{equation*}
\frac{\partial}{\partial h_i} K(\bm h^\ast)=2c h_i^\ast,\quad \forall i=1,2,\ldots,N.
\end{equation*}
Differentiating the effective precision $K(\bm h)$, whose calculation is deferred to \Cref{lem:precision common factor}, yields
\begin{equation}\label{eq:K derivative common factor}
\frac{\partial}{\partial h_i} K(\bm h)=\left (1 - \beta_i \frac{\sum_{j=1}^N h_j \beta_j}{(\sigma_\eta^2+\gamma)^{-1} + \sum_{j=1}^N h_j \beta_j^2}\right )^2,\quad \forall i=1,2,\ldots,N.
\end{equation}
Therefore, the equilibrium outcome $\bm h^\ast=(h_1^\ast,h_2^\ast,\ldots,h_N^\ast)$ must satisfy
\begin{equation*}
\left (1-\beta_i\frac{\sum_{j=1}^N h_j^\ast\beta_j}{(\sigma_\eta^2+\gamma)^{-1}+\sum_{j=1}^N h_j^\ast\beta_j^2}\right )^2=2ch_i^\ast,\quad \forall i=1,2,\ldots,N.
\end{equation*}

For the social planner who maximizes social welfare $W$ (recall \Cref{eq:social welfare}), namely
\begin{equation*}
W(\bm h)=\Pi(\bm h)+\sum_{i=1}^N U_i(h_i;p_i(\bm h))=K(\bm h)-\sum_{i=1}^N \frac c2 h_i^2,
\end{equation*}
the first-order condition gives
\begin{equation}\label{eq:FOC sp}
\frac{\partial}{\partial h_i} K(\bm h^{\text{sp}})=c h_i^{\text{sp}},\quad \forall i=1,2,\ldots,N.
\end{equation}
Equivalently, this suggests
\begin{equation*}
\left (1-\beta_i\frac{\sum_{j=1}^N h_j^{\text{sp}}\beta_j}{(\sigma_\eta^2+\gamma)^{-1}+\sum_{j=1}^N h_j^{\text{sp}}\beta_j^2}\right )^2=c h_i^{\text{sp}},\quad \forall i=1,2,\ldots,N.
\end{equation*}

By definition of $\Lambda(\bm h)$ in \Cref{eq:FOC lambda}, we derive
\begin{equation*}
h_i^\ast=\frac{1}{2c}\big (1-\beta_i \Lambda(\bm h^\ast)\big )^2,\quad h_i^{\text{sp}}=\frac{1}{c}\big (1-\beta_i \Lambda(\bm h^{\text{sp}})\big )^2,\quad \forall i=1,2,\ldots,N,
\end{equation*}
therefore proving \Cref{eq:effort using Lambda}. In any market with at least one positive $\beta_i$ (the other case is trivial since every creator has the same $\beta_i$), $\Lambda(\bm h^\ast)$ and $\Lambda(\bm h^{\text{sp}})$ are both strictly positive.
\end{proof}

\begin{lemma}[Effective Precision]\label{lem:precision common factor}
Under \Cref{eq:common factor,eq:Bayesian mu}, the precision $K(\bm h)$ induced by creators' effort vector $\bm h=(h_1,h_2,\ldots,h_N)'$ is:
\begin{equation*}
K(\bm h) = \sum_{j=1}^N h_j - \frac{(\sum_{j=1}^N h_j \beta_j)^2}{(\sigma_\eta^2+\gamma)^{-1} + \sum_{j=1}^N h_j \beta_j^2}.
\end{equation*}
\end{lemma}
\begin{proof}
The game proceeds as described in \Cref{def:seq of play}. We first take $h_i>0$ for all $i$; the final precision formula extends to zero effort by continuity. In the final stage of data aggregation, the buyer observes $\bm h$ and signals $\bm s$. From \Cref{eq:common factor,eq:Bayesian mu}, the signals are generated as
\begin{equation*}
s_i=X+\beta_i \eta+\varepsilon_i,\quad \varepsilon_i\sim \mathcal N(0,h_i^{-1}),\quad \mu_\eta\sim \mathcal N(0,\gamma),\quad \eta\mid \mu_\eta\sim \mathcal N(\mu_\eta,\sigma_\eta^2).
\end{equation*}
The noise vector $\bm \xi$ has mean $\bm 0$ and covariance matrix:
\begin{equation*}
\bm \Sigma = \E[\bm \xi \bm \xi'] = \bm D^{-1}+(\sigma_\eta^2+\gamma) \bm \beta \bm \beta',
\end{equation*}
where $\bm D = \operatorname{diag}(h_1, \dots, h_N)$ and $\bm \beta = (\beta_1, \dots, \beta_N)'$.
Given this Gaussian model, standard statistics \citep{aitken1936iv} gives the Best Linear Unbiased Estimator (BLUE) optimizing \Cref{eq:MSE} as
\begin{equation}\label{eq:BLUE Bayesian}
\hat X=(\bm 1' \bm \Sigma^{-1} \bm 1)^{-1} \bm 1' \bm \Sigma^{-1} \bm s,\quad \text{where }\bm \Sigma = \bm D^{-1} + (\sigma_\eta^2+\gamma) \bm \beta \bm \beta'.
\end{equation}
Using the {Sherman-Morrison-Woodbury identity} \citep{sherman1950adjustment,woodbury1950inverting},
\begin{equation*}
\bm \Sigma^{-1} = \left (\bm D^{-1} + \big ((\sigma_\eta^2+\gamma) \bm \beta\big ) \bm \beta'\right )^{-1}= \bm D - \frac{(\sigma_\eta^2+\gamma) \bm D \bm \beta \bm \beta' \bm D}{1 + (\sigma_\eta^2+\gamma) \bm \beta' \bm D \bm \beta}.
\end{equation*}
Substituting this into the precision formula $K(\bm h) = \bm 1' \bm \Sigma^{-1} \bm 1$, we derive
\begin{equation}\label{eq:precision general}
K(\bm h) = \bm 1' \left( \bm D - \frac{(\sigma_\eta^2+\gamma) \bm D \bm \beta \bm \beta' \bm D}{1 + (\sigma_\eta^2+\gamma) \bm \beta' \bm D \bm \beta} \right) \bm 1 = \sum_{i=1}^N h_i - \frac{(\sum_{i=1}^N h_i \beta_i)^2}{(\sigma_\eta^2+\gamma)^{-1} + \sum_{i=1}^N h_i \beta_i^2}.
\end{equation}
This completes the proof.
\end{proof}

\subsection{Underpowered Creative Incentives (Proof of \Cref{obs:noInnov})}
\begin{proof}[Proof of \Cref{obs:noInnov}]

To prove the uniqueness of $\bm h^\ast$ and $\bm h^{\text{sp}}$, we claim $K(\bm h)$ is concave. This is because the first term $\sum_{i=1}^N h_i$ is linear, and the second term $\big (\sum_{i=1}^N h_i \beta_i\big )^2\big /\big ((\sigma_\eta^2+\gamma)^{-1} + \sum_{i=1}^N h_i \beta_i^2\big )$, a quadratic-over-linear function in $\bm h$, is also convex \citep[p. 73]{boyd2004convex}.
Given the production cost $C(h_i)=\frac c2 h_i^2$ is strictly convex, both profit $\Pi$ and social welfare $W$, defined as
\begin{equation*}
\Pi(\bm h)=K(\bm h)-\sum_{i=1}^N ch_i^2,\quad W(\bm h)=K(\bm h)-\sum_{i=1}^N \frac c2 h_i^2
\end{equation*}
are strictly concave in $\bm h$. Therefore, the market equilibrium $\bm h^\ast$ maximizing $\Pi$ and the social optimum $\bm h^{\text{sp}}$ maximizing $W$ are both unique, as claimed.

We now compare $K(\bm h^\ast)$ to $K(\bm h^{\text{sp}})$. For any $y\ge 0$, we define the minimal total production cost required to attain an overall precision $y$ as $C_{\min}(y)$, namely
\begin{equation*}
C_{\min}(y):=\min_{\bm h\in \mathbb R_{\ge 0}^N} \frac c2 \sum_{i=1}^N h_i^2\text{ s.t. }K(\bm h)\ge y.
\end{equation*}

Since $K(\bm h)$ is strictly concave and $\sum_{i=1}^N h_i^2$ is strictly convex, $C_{\min}(y)$ is convex. We thus rewrite the buyer's and the social planner's optimization problems in w.r.t. $\Pi$ and $W$ as:
\begin{align*}
\max_{\bm h} \Pi(\bm h)&=\max_{\bm h} \left (K(\bm h)-c\sum_{i=1}^N h_i^2\right )=\max_{y} \left (y-2C_{\min}(y)\right ),\\
\max_{\bm h} W(\bm h)&=\max_{\bm h} \left (K(\bm h)-\frac c2\sum_{i=1}^N h_i^2\right )=\max_{y} \left (y-C_{\min}(y)\right ).
\end{align*}
The first-order conditions reveal $C_{\min}'(y^\ast)=\frac 12$ and $C_{\min}'(y^{\text{sp}})=1$. Since $C_{\min}$ is strictly convex, $C_{\min}'$ is strictly increasing and thus $y^\ast<y^{\text{sp}}$. This implies $K(\bm h^\ast)=y^\ast<y^{\text{sp}}=K(\bm h^{\text{sp}})$.
\end{proof}

\begin{remark}[Individual Creators May Over-Produce]\label{remark:individual creator over-produce}
\Cref{obs:noInnov} compares the overall effective precision $K$ under $\bm h^\ast$ and $\bm h^{\text{sp}}$. The stronger per-creator version---namely $h_i^\ast\le h_i^{\text{sp}}$ for all $i$---is incorrect: Consider a 2-creator example with $\beta_1=1$ and $\beta_2=2$ (hence creator $2$ is highly imitative), production cost coefficient $c=1$, and $(\sigma_\eta^2+\gamma)^{-1}=\frac 14$. 
We first solve the social planner solution $\bm h^{\text{sp}}$: According to \Cref{eq:FOC lambda,eq:effort using Lambda}, $\bm h^{\text{sp}}$ (which is unique according to \Cref{obs:noInnov}) satisfies
\begin{equation*}
\Lambda(\bm h^{\text{sp}})=\frac{h_1^{\text{sp}}+2h_2^{\text{sp}}}{\frac 14+h_1^{\text{sp}}+4h_2^{\text{sp}}},\quad h_1^{\text{sp}}=(1-\Lambda(\bm h^{\text{sp}}))^2,\quad h_2^{\text{sp}}=(1-2\Lambda(\bm h^{\text{sp}}))^2.
\end{equation*}
Solving it gives $h_1^{\text{sp}}=\frac 14$, $h_2^{\text{sp}}=0$, and $\Lambda(\bm h^{\text{sp}})=\frac 12$.
On the other hand, $\bm h^\ast$ (also unique) satisfies
\begin{equation*}
\Lambda(\bm h^\ast)=\frac{h_1^\ast+2h_2^\ast}{\frac 14+h_1^\ast+4h_2^\ast},\quad h_1^\ast=\frac 12(1-\Lambda(\bm h^\ast))^2,\quad h_2^\ast=\frac 12(1-2\Lambda(\bm h^\ast))^2.
\end{equation*}
Prove by contradiction: Suppose that $h_i^\ast\le h_i^{\text{sp}}$ for all $i$, which implies $h_2^\ast=0$ (since efforts are non-negative), we must have $\Lambda(\bm h^\ast)=\frac 12$, which equals $\Lambda(\bm h^{\text{sp}})$. Since $\Lambda(\bm h)$ is monotone in $h_1$ once fixing $h_2$, $h_1^\ast=h_1^{\text{sp}}=\frac 14$. But $\frac 14\ne \frac 12(1-\frac 12)^2$. This contradicts $h_2^\ast=0$ and hence implies $h_2^\ast>0=h_2^{\text{sp}}$.

We further remark that, even if $h_2^{\text{sp}}>0$, it is still possible that $h_2^\ast>h_2^{\text{sp}}$: Consider the same example but now $\beta_2=1.7$, that is, creator $2$ remains highly imitative but to a less extent. Numerically solving \Cref{eq:FOC lambda,eq:effort using Lambda}, we have $\bm h^{\text{sp}}\approx (0.245,0.020)$ and $\bm h^\ast\approx (0.162,0.036)$. In either case, the intuition is the same: Creator 2 should invest no or little effort in the social optimum; but in the equilibrium, due to creator 1's under-investment, their effort becomes (slightly) valuable.
\end{remark}

\subsection{Originality Penalty (Proof of \Cref{obs:biasOrigin})}
\begin{proof}[Proof of \Cref{obs:biasOrigin}]
To derive the monotonicity of $R_i$ in $\beta_i$, we first prove that $\Lambda(\bm h^\ast)<\Lambda(\bm h^{\text{sp}})$ (where $\Lambda(\bm h)$ is defined in \Cref{eq:FOC lambda}). Plugging \Cref{eq:effort using Lambda} into the definition of $\Lambda(\bm h^\ast)$,
\begin{align*}
&\qquad (\sigma_\eta^2+\gamma)^{-1} \Lambda(\bm h^\ast)+\sum_{j=1}^N h_j^\ast \beta_j^2\Lambda(\bm h^\ast)=\sum_{j=1}^N h_j^\ast \beta_j ,\\
&\Longrightarrow 2c (\sigma_\eta^2+\gamma)^{-1} \Lambda(\bm h^\ast)=\sum_{j=1}^N \big (1-\beta_j \Lambda(\bm h^\ast)\big )^2 \beta_j\left (1-\beta_j \Lambda(\bm h^\ast)\right )=\sum_{j=1}^N \beta_j \big (1-\beta_j \Lambda(\bm h^\ast)\big )^3,
\end{align*}
and similarly plugging \Cref{eq:effort using Lambda} into $\Lambda(\bm h^{\text{sp}})$ gives
\begin{equation*}
c(\sigma_\eta^2+\gamma)^{-1} \Lambda(\bm h^{\text{sp}})=\sum_{j=1}^N \beta_j \big (1-\beta_j \Lambda(\bm h^{\text{sp}})\big )^3.
\end{equation*}
Viewing both equations as a function of $\Lambda$, the LHS's are both linear in $\Lambda$. Meanwhile, the RHS is monotonically decreasing in $\Lambda>0$ (recall that $\bm \beta$ is non-negative and at least one $\beta_i>0$) since:
\begin{equation*}
\frac{\mathrm d}{\mathrm d \Lambda} \sum_{j=1}^N \beta_j \big (1-\beta_j \Lambda\big )^3=-3\sum_{j=1}^N \beta_j^2 \big (1-\beta_j \Lambda\big )^2<0.
\end{equation*}
Therefore, we obtain $\Lambda(\bm h^\ast)<\Lambda(\bm h^{\text{sp}})$. Plug \Cref{eq:effort using Lambda} into the definition of $R_i=h_i^\ast/h_i^{\text{sp}}$:
\begin{equation*}
R_i=\frac{h_i^\ast}{h_i^{\text{sp}}}=\frac{\frac{1}{2c}\big (1-\beta_i \Lambda(\bm h^\ast)\big )^2}{\frac{1}{c}\big (1-\beta_i \Lambda(\bm h^{\text{sp}})\big )^2}=\frac 12 \left (\frac{1-\beta_i \Lambda(\bm h^\ast)}{1-\beta_i \Lambda(\bm h^{\text{sp}})}\right )^2.
\end{equation*}

We now view this ratio $R_i$ as a function of $\beta_i$ (in this fixed market, i.e., we fix $\Lambda(\bm h^\ast)$ and $\Lambda(\bm h^{\text{sp}})$, and consider the $\beta_i$'s corresponding to different creators). Formally, let $g(\beta):=\frac{1-\beta \Lambda(\bm h^\ast)}{1-\beta \Lambda(\bm h^{\text{sp}})}$ and $R(\beta):=\frac 12g(\beta)^2$, so $R_i=R(\beta_i)$ for all $i=1,2,\ldots,N$. Since we know $\Lambda(\bm h^\ast)<\Lambda(\bm h^{\text{sp}})$,
\begin{equation}\label{eq:non-redundant}
0< 1-\beta \Lambda(\bm h^{\text{sp}})\le 1-\beta \Lambda(\bm h^\ast),\quad \forall 0\le \beta < 1/\Lambda(\bm h^{\text{sp}}),
\end{equation}
hence we know $g(\beta)>0$ for all $\beta<1/\Lambda(\bm h^{\text{sp}})$. Moreover, we also have
\begin{equation*}
g'(\beta)=\frac{\mathrm d}{\mathrm d \beta} \frac{1-\beta \Lambda(\bm h^\ast)}{1-\beta \Lambda(\bm h^{\text{sp}})}=\frac{\Lambda(\bm h^{\text{sp}})-\Lambda(\bm h^\ast)}{ (1-\beta \Lambda(\bm h^{\text{sp}}))^2}>0,\quad \forall \beta,
\end{equation*}
where the last inequality again uses $\Lambda(\bm h^\ast)<\Lambda(\bm h^{\text{sp}})$. Henceforth $R'(\beta)=g(\beta)g'(\beta)>0$ for all $\beta\in [0,1/\Lambda(\bm h^{\text{sp}}))$.
According to the condition that $\beta_i<1/\Lambda(\bm h^{\text{sp}})$ for all $i$, $R_i$ is increasing in $\beta_i$.

It only remains to lower and upper bound $R_i$. Since $\beta_i\ge 0$ and $R_i$ is monotone in $\beta_i$, we have
\begin{equation*}
R_i\ge \frac 12 \left (\frac{1-0 \Lambda(\bm h^\ast)}{1-0 \Lambda(\bm h^{\text{sp}})}\right )^2=\frac 12,\quad \text{with equality attained when }\beta_i=0.
\end{equation*}
Furthermore, in the limiting case where $N\to \infty$ and all creators have the same $\beta_i\equiv \beta>0$, each creator must have a symmetric $h^\ast$ or $h^{\text{sp}}$. In this case, the $\Lambda(\bm h)$ defined in \Cref{eq:FOC lambda} becomes
\begin{equation*}
\Lambda(h \bm 1)=\frac{N h \beta}{(\sigma_\eta^2+\gamma)^{-1}+Nh \beta^2},\quad \forall h\ge 0.
\end{equation*}
Therefore \Cref{eq:effort using Lambda} implies $\bm h^\ast=h^\ast\bm 1$ and $\bm h^{\text{sp}}=h^{\text{sp}}\bm 1$ for scalars $h^\ast\ge 0$ and $h^{\text{sp}}\ge 0$, such that
\begin{equation*}
h^\ast=\frac{1}{2c} \left (1-\frac{N h^\ast \beta^2}{(\sigma_\eta^2+\gamma)^{-1}+Nh^\ast \beta^2}\right )^2,\quad h^{\text{sp}}=\frac{1}{c} \left (1-\frac{N h^{\text{sp}} \beta^2}{(\sigma_\eta^2+\gamma)^{-1}+Nh^{\text{sp}} \beta^2}\right )^2.
\end{equation*}
When $N\to \infty$, the $(\sigma_\eta^2+\gamma)^{-1}$ terms are dominated. Therefore, asymptotically we have
\begin{equation*}
h^\ast\sim\sqrt[3]{\frac{1}{2cN^2 \beta^4(\sigma_\eta^2 + \gamma)^2}},\quad 
h^{\text{sp}}\sim\sqrt[3]{\frac{1}{cN^2 \beta^4(\sigma_\eta^2 + \gamma)^2}}.
\end{equation*}
This gives $\lim_{N\to \infty} R_i=\lim_{N\to \infty} \frac{h^\ast}{h^{\text{sp}}}=\frac{1}{\sqrt[3]{2}}$ for all $i=1,2,\ldots,N$, as claimed.
\end{proof}

\begin{remark}[No Near-Redundant Creator]\label{remark:non-redundant}
{We now discuss the condition that $\beta_i<1/\Lambda(\bm h^{\text{sp}})$ for all $i$ (see \Cref{footnote:non-redundant} of \Cref{obs:biasOrigin}). We first discuss its two implications, which is the reason why we call it ``a market without near-redundant creator'' in \Cref{footnote:non-redundant}:
\begin{enumerate}
\item \textbf{BLUE $\hat X$ places all positive weights.} Recall from the proof of \Cref{lem:precision common factor} that $\bm \Sigma^{-1}=\bm D-\frac{(\sigma_\eta^2+\gamma)\bm D\bm\beta\bm\beta'\bm D}{1+(\sigma_\eta^2+\gamma)\bm\beta'\bm D\bm\beta}$. Hence the weight $\hat X=(\bm 1'\bm\Sigma^{-1}\bm 1)^{-1}\bm 1'\bm\Sigma^{-1}\bm s$ assigns to creator $i$ is
\begin{equation*}
w_i\propto [\bm\Sigma^{-1}\bm 1]_i=h_i-\frac{(\sigma_\eta^2+\gamma)h_i\beta_i\sum_{j=1}^N h_j\beta_j}{1+(\sigma_\eta^2+\gamma)\sum_{j=1}^N h_j\beta_j^2}=h_i\big (1-\beta_i \Lambda(\bm h)\big ),\quad \forall i=1,2,\ldots,N,
\end{equation*}
which is positive when $\beta_i<1/\Lambda(\bm h)$. Imposed at $\bm h^{\text{sp}}$, the condition thus makes the optimal estimator weight every creator positively. Should it be violated---thus the BLUE puts a negative weight---this means their highly correlated signal is only used to cancel out the common factor $\eta$, hence the name ``near-redundant.''

\item \textbf{Socially optimum is monotone in correlation.} By \Cref{eq:effort using Lambda}, $h_i^{\text{sp}}=\frac 1c(1-\beta_i \Lambda(\bm h^{\text{sp}}))^2$. On the range $\beta_i\in [0,1/\Lambda(\bm h^{\text{sp}}))$, $1-\beta_i \Lambda(\bm h^{\text{sp}})$ is positive and decreasing in $\beta_i$, which means $h_i^{\text{sp}}$ is decreasing in $\beta_i$: the planner asks original creators to invest more. Past the threshold this reverses---$h_i^{\text{sp}}$ turns back up---and the reason mirrors the previous item: the near-redundant creator becomes valuable again, but purely as a hedge against the common $\eta$.
\end{enumerate}

Both readings carry over to the market equilibrium $\bm h^\ast$: Since $\Lambda(\bm h^\ast)<\Lambda(\bm h^{\text{sp}})$ (established in the proof of \Cref{obs:biasOrigin}), $\beta_i<1/\Lambda(\bm h^{\text{sp}})$ implies $\beta_i<1/\Lambda(\bm h^\ast)$, so $\bm h^\ast$ likewise features positive weights and $\beta$-monotone efforts.
We then state the technical impact of this condition:
\begin{enumerate}
\item \textbf{Role in the proof of \Cref{obs:biasOrigin}.} This condition ensures \Cref{eq:non-redundant}, which says $g(\beta)=\frac{1-\beta \Lambda(\bm h^\ast)}{1-\beta \Lambda(\bm h^{\text{sp}})}$ is positive at all $\beta_1,\beta_2,\ldots,\beta_N$. Combined with the subsequent formula that $g'(\beta)>0$ for all $\beta$, this implies the monotonicity of $R(\beta)=\frac 12 g(\beta)^2$. Note that, one may think the facts that $g(0)=\frac 11=1$ and $g'(\beta)>0$ suffice to ensure the monotonicity of $g(\beta)$ and consequently $R(\beta)$. This is untrue, because at $\beta=1/\Lambda(\bm h^{\text{sp}})$, the function $g(\beta)$ flips from $+\infty$ to $-\infty$ (similar to $\tan x$: non-monotone but has an always positive derivative $\sec^2 x$).
\item \textbf{\Cref{obs:biasOrigin} Fails without it.} Consider the same example in \Cref{remark:individual creator over-produce}, but now $\beta_2=2.2$ (i.e., even more redundant). Now $\bm h^{\text{sp}}\approx (0.252,0.009)$, which implies $\beta_2 \Lambda(\bm h^{\text{sp}})=1.096>1$. In this case, we have $\bm h^\ast\approx (0.173,0.004)$, hence $R_1=0.687$ and $R_2=0.490$---both the monotonicity (implying $R_2>R_1$) and the lower bound (implying $R_2>0.5$) fail.
\item \textbf{Per-creator \Cref{obs:noInnov} fails even with it.} One may naturally ask: Given this condition, would the per-creator under-investment variant of \Cref{obs:noInnov}---discussed in \Cref{footnote:individual creator over-produce} and \Cref{remark:individual creator over-produce}---now be true? The answer remains negative. In the second example of \Cref{remark:individual creator over-produce} (with $\beta_2=1.7$), we have $\bm h^{\text{sp}}\approx (0.245,0.020)$ and hence $\Lambda(\bm h^{\text{sp}})\approx 0.505$. In this case $\beta_2 \Lambda(\bm h^{\text{sp}})\approx 0.858<1$, but we still have $h_2^\ast>h_2^{\text{sp}}$. Should we really want a per-creator under-investment result, we would need $\beta_i\le \frac{\sqrt 2-1}{\sqrt 2\Lambda(\bm h^{\text{sp}})-\Lambda(\bm h^\ast)}$ for all $i$ (a sufficient condition is $\beta_i\le (1-\frac{1}{\sqrt 2})/\Lambda(\bm h^{\text{sp}})\approx 0.293/\Lambda(\bm h^{\text{sp}})$, which is strictly stronger than $\beta_i<1/\Lambda(\bm h^{\text{sp}})$).
\end{enumerate}
}
\end{remark}

\subsection{Minimax Robust Optimization View of \Cref{eq:Bayesian mu}}
\begin{lemma}[Minimax Robust Optimization View]\label{lem:minimax}
The effective precision $K(\bm h)$ derived under the Bayesian setting in \Cref{eq:Bayesian mu} (where $\mu_\eta \sim \mathcal N(0, \gamma)$ and $\eta\mid \mu_\eta\sim \mathcal N(\mu_\eta,\sigma_\eta^2)$) is identical to that derived from a {Minimax Robust Optimization} framework where the buyer only knows $\mu_\eta^2\le \gamma$.
Formally, let the buyer seek a linear estimator $\hat X = \bm w' \bm s$ to minimize the worst-case MSE:
\begin{equation}\label{eq:minimax view}
\min_{\bm w} \max_{\mu_\eta^2\le \gamma}\E[(\bm w'\bm s-X)^2],\quad \text{s.t. }\bm w'\bm 1=1.
\end{equation}
The solution $\bm w^\ast$ is identical to BLUE weights in the Bayesian setting, and the MSE equals $K(\bm h)^{-1}$.
\end{lemma}
\begin{proof}
We decompose the MSE of any linear estimator $\hat X = \bm w' \bm s$ into variance and bias:
\begin{equation*}
  \text{MSE}(\hat X) = \underbrace{\Var(\hat X)}_{\text{Variance}} + \underbrace{(\E[\hat X] - X)^2}_{\text{Squared Bias}}.
\end{equation*}
Under the robust optimization view of \Cref{eq:common factor}, $\bm s = \bm 1 X + \bm \beta \eta + \bm \varepsilon$, where $\bm\varepsilon=(\varepsilon_1,\ldots,\varepsilon_N)'$, $\E[\eta] = \mu_\eta$, and $\mu_\eta^2\le \gamma$.
The variance depends only on the fluctuations around the mean, not on $\mu_\eta$ itself. Therefore, letting $\bm D=\operatorname{diag}(h_1,\ldots,h_N)$ and $\bm \Sigma_{\text{noise}} = \bm D^{-1} + \sigma_\eta^2 \bm \beta \bm \beta'$, we know $\Var(\hat X) = \bm w' \bm \Sigma_{\text{noise}} \bm w$ is independent of $\mu_\eta$.
On the other hand, since $\bm w' \bm 1 = 1$, the mean of $\hat X$ is $\E[\hat X] = \bm w' (\bm 1 X + \bm \beta \mu_\eta) = X + \mu_\eta (\bm w' \bm \beta)$. The squared bias term therefore equals $\mu_\eta^2 (\bm w' \bm \beta)^2$.

This reveals that the inner maximization w.r.t. $\mu_\eta$ in \Cref{eq:minimax view} is equivalent to $\max_{\mu_\eta^2 \le \gamma} \mu_\eta^2 (\bm w' \bm \beta)^2$. The maximum is clearly attained at the boundary $\mu_\eta^2 = \gamma$.
Thus, the buyer's problem in \Cref{eq:minimax view} simplifies to:
\begin{equation*}
    \min_{\bm w} \left( \bm w' \bm \Sigma_{\text{noise}} \bm w + \gamma (\bm w' \bm \beta)^2 \right) \quad \text{s.t. } \bm w' \bm 1 = 1.
\end{equation*}
Notice that $\gamma (\bm w' \bm \beta)^2 = \bm w' (\gamma \bm \beta \bm \beta') \bm w$. We can combine the quadratic terms:
\begin{equation*}
    \min_{\bm w} \bm w' \left( \bm \Sigma_{\text{noise}} + \gamma \bm \beta \bm \beta' \right) \bm w \quad \text{s.t. } \bm w' \bm 1 = 1.
\end{equation*}
  Let $\bm \Sigma = \bm \Sigma_{\text{noise}} + \gamma \bm \beta \bm \beta' = \bm D^{-1} + (\sigma_\eta^2 + \gamma) \bm \beta \bm \beta'$.
The Lagrangian for this problem is:
\begin{equation*}
    \mathcal{L}(\bm w, \nu) = \frac{1}{2} \bm w' \bm \Sigma \bm w - \nu (\bm w' \bm 1 - 1).
\end{equation*}
  The First-Order Condition (FOC) with respect to $\bm w$ reveals $\bm \Sigma \bm w^\ast = \nu \bm 1$, or equivalently, $\bm w^\ast = \nu \bm \Sigma^{-1} \bm 1$.
  Since $\bm 1' \bm w^\ast = 1$, we know $\nu = (\bm 1' \bm \Sigma^{-1} \bm 1)^{-1}$ and therefore:
\begin{equation*}
    \bm w^\ast = (\bm 1' \bm \Sigma^{-1} \bm 1)^{-1} \bm \Sigma^{-1} \bm 1.
\end{equation*}
  Since this is the covariance matrix $\bm \Sigma$ in the Bayesian setting (see the proof of \Cref{lem:precision common factor}), this is exactly the formula for BLUE weights in the Bayesian setting.
  Furthermore, the worst-case MSE induced by $\bm w^\ast$ is
\begin{equation*}
    \text{MSE}^\ast = \bm w^{\ast\prime} \bm \Sigma \bm w^\ast = \left( \nu \bm 1' \bm \Sigma^{-1} \right) \bm \Sigma \left( \nu \bm \Sigma^{-1} \bm 1 \right) = (\bm 1' \bm \Sigma^{-1} \bm 1)^{-1},
\end{equation*}
  where we used $\nu = (\bm 1' \bm \Sigma^{-1} \bm 1)^{-1}$.
In the Bayesian setting, we have $K(\bm h)=\bm 1'\bm \Sigma^{-1}\bm 1$; therefore, $\text{MSE}^\ast=K(\bm h)^{-1}$, as claimed.
\end{proof}

\section{Proofs for \Cref{sec:dynamic}}\label{app:dynamics}

In this section, we first prove that the assumption on supply elasticity (\Cref{eq:elasticity dynamic}) is equivalent to the monotonicity of rank-and-file creator-innovator ratio $\rho(p)$. We then prove \Cref{obs:price trap,obs:collapse}, which together constitute for our ``curse of precision'' observation.

\subsection{Supply Elasticity and Monotonicity of Ratio}
\begin{lemma}[Supply Elasticity and Monotonicity of Ratio]\label{lem:rho dynamic}
The condition in \Cref{eq:elasticity dynamic}, namely $\mathcal E_C(p)>\mathcal E_O(p)$ for all $p>0$ where $\mathcal E_j(p):=\frac{\mathrm d\ln N_j(p)}{\mathrm d\ln p}$ ($\forall j\in \{O,C\}$), is equivalent to
\begin{equation*}
\rho'(p)>0,~\forall p>0,\quad \text{where }\rho(p):=\frac{N_C(p)}{N_O(p)}.
\end{equation*}
\end{lemma}
\begin{proof}
The condition that $\rho'(p)>0$ is equivalent to
\begin{equation*}
0<\frac{\mathrm d}{\mathrm d p}\frac{N_C(p)}{N_O(p)}=\frac{N_C'(p) N_O(p)-N_C(p) N_O'(p)}{\big (N_O(p)\big )^2},
\end{equation*}
which is equivalent to $\frac{N_C'(p)}{N_C(p)}>\frac{N_O'(p)}{N_O(p)}$. Multiplying $p$ onto both sides, we get $\mathcal E_C(p)>\mathcal E_O(p)$.
\end{proof}

\subsection{Upper Bound on Market Price (Proof of \Cref{obs:price trap})}
\begin{proof}[Proof of \Cref{obs:price trap}]
According to the buyer's optimization \Cref{eq:profit dynamic}, at any $t$, the market price $p(t)$ maximizes the instantaneous profit $\Pi_{\text{inst}}(t)$ (we omit all dependencies on $t$ for readability):
Let $a:=\frac{\partial K}{\partial S_O}$ and $b:=\frac{\partial K}{\partial S_C}$. By \Cref{lem:precision vs data dynamic}, $0<a\le 1$ and $0<b\le \kappa_C$.
\begin{align*}
&\quad \Pi_{\text{inst}}={\frac{\mathrm dK}{\mathrm d t}}-{\Big (N_O(p)+N_C(p)\Big ) p h(p)}\\
&\overset{(a)}{=}\frac{\partial K}{\partial S_O} \frac{\mathrm d S_O}{\mathrm d t}+\frac{\partial K}{\partial S_C} \frac{\mathrm d S_C}{\mathrm d t}-\frac{p^2}{c} \Big (N_O(p)+N_C(p)\Big )\\
&\overset{(b)}{=}a \left( N_O \frac{p}{c} - \delta S_O \right) + b \left( N_C \frac{p}{c} - \delta S_C \right) - \frac{p^2}{c} (N_O + N_C)\\
&=\frac pc \Big (N_O (a-p)+N_C(b-p)\Big )-\delta \big (aS_O+bS_C\big ),
\end{align*}
where (a) uses the Law of Total Derivatives, and (b) uses the best effort formula $h(p)=p/c$ and the dynamics of $S_C(t)$ and $S_O(t)$ in \Cref{eq:S dynamic}.

Therefore, the instantaneous profit $\Pi_{\text{inst}}(t)$ decomposes into a variable component (depending on $p(t)$) and a sunk cost component (the $-\delta(aS_O+bS_C)$ term, which arises due to depreciation).
For the buyer to sustain a price $p(t)$, the variable component must be non-negative (otherwise the buyer would be better off by setting $p(t)=0$ and bearing only the sunk cost). Thus, using $a\le 1$ and $b\le \kappa_C$, the market price $p(t)$ must satisfy (again omitting all $t$'s for readability):
\begin{equation*}
N_O(p) (1 - p) + N_C(p) (\kappa_C - p) \ge N_O(p) (a - p) + N_C(p) (b - p)\ge 0.
\end{equation*}
Rearranging for $p$ gives $p\le \kappa_C+\frac{1-\kappa_C}{1+\rho(p)}$, where we recall $\rho(p)=\frac{N_C(p)}{N_O(p)}$ is the rank-and-file creator-innovator ratio. From \Cref{eq:elasticity dynamic}, $\rho(p)$ is strictly increasing in $p$, and thus the RHS is decreasing in $p$. Since the LHS is increasing, the fixed-point equation has a unique solution, namely $\bar p(\kappa_C)$, such that
\begin{equation}\label{eq:p bar dynamic}
\bar p(\kappa_C)=\kappa_C+\frac{1-\kappa_C}{1+\rho\big (\bar p(\kappa_C)\big )},\quad \forall \kappa_C\in (0,1].
\end{equation}
Again invoking the Implicit Function Theorem to \Cref{eq:p bar dynamic} (whose detailed calculation is deferred to \Cref{lem:p bar monotone}), we are able to prove $\bar p'(\kappa_C)>0$, i.e., $\bar p$ is increasing in $\kappa_C$. This further establishes the continuity of $\bar p$, which gives
\begin{equation*}
\lim_{\kappa_C\to 0^+} \bar p(\kappa_C)=\bar p(0),\quad \text{where }\bar p(0)=0+\frac{1-0}{1+\rho(\bar p(0))}=\left (1+\frac{N_C(\bar p(0))}{N_O(\bar p(0))}\right )^{-1}.
\end{equation*}
This matches the definition of $p^\ast$ in \Cref{eq:price trap}. Hence, $\lim_{\kappa_C\to 0^+} \bar p(\kappa_C)=p^\ast$, as claimed.
\end{proof}

\begin{lemma}[{Marginal Effective Precision}]\label{lem:precision vs data dynamic}
The effective precision $K(t)$ defined in \Cref{eq:precision dynamic} satisfies
\begin{equation*}
0<\frac{\partial K}{\partial S_O}(t)\le 1,\qquad 0<\frac{\partial K}{\partial S_C}(t)\le \kappa_C(t),\quad \forall t\ge 0,
\end{equation*}
where $\kappa_C(t)$ is defined in \Cref{eq:kappa_C dynamic}, which we recall as
\begin{equation*}
\kappa_C(t):=\left (1+\big (\sigma_\eta^2+\lambda K(t)\big ) S_C(t) \beta_C^2\right )^{-2}.
\end{equation*}
\end{lemma}
\begin{proof}
From \Cref{eq:precision dynamic}, the effective precision $K(t)$ is given by the implicit equation:
\begin{equation*}
\mathcal F(K, S_O, S_C) := K - S_O - \frac{S_C}{1 + \big (\sigma_\eta^2+\lambda K\big ) S_C \beta_C^2} = 0,
\end{equation*}
By the Implicit Function Theorem \citep{krantz2002implicit}, the partial derivatives of $K$ w.r.t. stocks $S_O$ and $S_C$ are therefore given by
\begin{equation*}
\frac{\partial K}{\partial S_j} = - \frac{\partial \mathcal F / \partial S_j}{\partial \mathcal F / \partial K},\quad j\in \{O,C\}.
\end{equation*}
Calculating the partial derivatives of $\mathcal F$ with respect to $S_O$, $S_C$, and $K$ gives:
\begin{equation*}\begin{aligned}
\frac{\partial \mathcal F}{\partial S_O} &= -1, \\
\frac{\partial \mathcal F}{\partial S_C} &= - \frac{1 + \big (\sigma_\eta^2 +\lambda K\big ) S_C \beta_C^2 - S_C\cdot \big (\sigma_\eta^2 +\lambda K\big ) \beta_C^2}{\big (1 + \big (\sigma_\eta^2 +\lambda K\big ) S_C \beta_C^2\big )^2} = - \frac{1}{\big (1 + \big (\sigma_\eta^2 +\lambda K\big ) S_C \beta_C^2\big )^2},\\
\frac{\partial \mathcal F}{\partial K} &= 1 - S_C \cdot \frac{\partial}{\partial K} \left( \frac{1}{1 + \big (\sigma_\eta^2 +\lambda K\big )S_C \beta_C^2} \right).
\end{aligned}\end{equation*}
As $(1 + (\sigma_\eta^2 +\lambda K )S_C \beta_C^2)^{-1}$ is decreasing in $K$, its derivative is non-positive. Hence $\frac{\partial \mathcal F}{\partial K} \ge 1$, and
\begin{equation*}\begin{aligned}
0<\frac{\partial K}{\partial S_O} = \frac{1}{\partial \mathcal F / \partial K} \le 1,\quad 
0<\frac{\partial K}{\partial S_C} = \frac{1}{\big (1 + \big (\sigma_\eta^2 +\lambda K\big ) S_C \beta_C^2\big )^2} \frac{1}{\partial \mathcal F / \partial K}\le \frac{1}{\big (1 + \big (\sigma_\eta^2 +\lambda K\big ) S_C \beta_C^2\big )^2}.
\end{aligned}\end{equation*}
Plugging in the definition of $\kappa_C(t)$ gives the conclusion.
\end{proof}

\begin{lemma}[Monotonicity of Price Upper Bound]\label{lem:p bar monotone}
The price upper bound $\bar p(\kappa_C)$ derived in \Cref{eq:p bar dynamic}, namely
\begin{equation*}
\bar p(\kappa_C)=\kappa_C+\frac{1-\kappa_C}{1+\rho\big (\bar p(\kappa_C)\big )},\quad \text{where }\rho(p)=\frac{N_C(p)}{N_O(p)},
\end{equation*}
is strictly increasing in $\kappa_C$, i.e., $\bar p'(\kappa_C)>0$ for all $\kappa_C\in (0,1]$.
\end{lemma}
\begin{proof}
From \Cref{eq:p bar dynamic}, the price upper bound $\bar p(\kappa_C)$ is given by the implicit equation:
\begin{equation*}
G(\bar p,\kappa_C):=\bar p-\kappa_C-\frac{1-\kappa_C}{1+\rho(\bar p)}=0.
\end{equation*}
By the Implicit Function Theorem \citep{krantz2002implicit}, we have
\begin{equation*}
\frac{\mathrm d \bar p}{\mathrm d \kappa_C}=-\frac{\partial G/\partial \kappa_C}{\partial G/\partial \bar p}=\frac{1-\frac{1}{1+\rho(\bar p)}}{1+\frac{1-\kappa_C}{(1+\rho(\bar p))^2} \rho'(\bar p)}.
\end{equation*}
The numerator is positive because $\rho(p)>0$ for all $p$, and the denominator is also positive because $\kappa_C\in (0,1]$ and also $\rho'(p)>0$ for all $p$ (from \Cref{eq:elasticity dynamic} and \Cref{lem:rho dynamic}). Therefore, $\bar p'(\kappa_C)>0$ for all $\kappa_C\in (0,1]$, as claimed.
\end{proof}

\subsection{Barrier of Model Precision (Proof of \Cref{obs:collapse})}
\begin{proof}[Proof of \Cref{obs:collapse}]
Recall the fixed-point definition for precision $K(t)$ in \Cref{eq:precision dynamic}, namely $K(t)=S_O(t)+\frac{S_C(t)}{1+ (\sigma_\eta^2 + \lambda K(t) )S_C(t)\beta_C^2}$. When innovators have exited ($S_O(t) \to 0$), this equation becomes
\begin{equation*}
K(t) = \frac{S_C(t)}{1 + \big (\sigma_\eta^2 + \lambda K(t)\big ) S_C(t) \beta_C^2 }.
\end{equation*}
In the limit of infinite correlated data (i.e., $S_C(t) \to \infty$), we have
\begin{equation*}
1=\lim_{S_C(t) \to \infty} K(t) \Big (S_C(t)^{-1} + \big (\sigma_\eta^2 + \lambda K(t)\big ) \beta_C^2\Big ) = \lim_{S_C(t) \to \infty} K(t) \big (\sigma_\eta^2 + \lambda K(t)\big ) \beta_C^2.
\end{equation*}
Therefore, the limiting precision $\bar K$ must satisfy
\begin{equation*}
\lambda \beta_C^2 \big (\bar K\big )^2 + \sigma_\eta^2 \beta_C^2 \bar K - 1 = 0.
\end{equation*}
Solving for the unique positive root yields the expression in \Cref{eq:collapse}.
\end{proof}

\section{Proofs for \Cref{sec:mechanism}}\label{sec:mechanism appendix}

In this section, we prove \Cref{lem:IC,obs:linear pricing}, and \Cref{obs:two part}.

\subsection{Linear Pricing Fails on Participation (Proof of \Cref{obs:linear pricing})}

\begin{proof}[Proof of \Cref{lem:IC}]
For each creator $i=1,2,\ldots,N$, under the linear pricing rule (or the two-part tariff rule considered later in \Cref{obs:two part}), their payoff $U_i(\bm h;\bm P)$ only depends on $h_i$ and $p_i$ (and also $B_i$ if considering the two-part tariff rule defined in \Cref{eq:two-part tariff}), namely $U_i(\bm h;\bm P)=p_ih_i+B_i-C(h_i)$. Thus the first-order condition is simply
\begin{equation*}
p_i=ch_i.
\end{equation*}
Therefore, the best response is $h_i^{\text{sp}}$---thus satisfying Incentive Compatibility \eqref{eq:IC}---if and only if $p_i=ch_i^{\text{sp}}$. This gives the only linear $\bm P$ attaining IC, i.e., $P_i(\bm h)=ch_i^{\text{sp}}\times h_i$, $\forall i=1,2,\ldots,N$.
\end{proof}
\begin{proof}[Proof of \Cref{obs:linear pricing}]
According to \Cref{lem:IC}, the $\bm P$ is uniquely determined by \Cref{eq:IC}.
We further pin down $T$ using Budget Balance \eqref{eq:BB}. This gives
\begin{equation*}
P_i(\bm h)=ch_i^{\text{sp}}\times h_i,~\forall i=1,2,\ldots,N;\quad T(\bm h)=F + \sum_{i=1}^N p_i h_i.
\end{equation*}
This gives the unique $\langle \bm P,T\rangle$ ensuring \Cref{eq:IC,eq:BB}. We write the payoff of each creator:
\begin{equation*}
U_i(\bm h^{\text{sp}};\bm P)=p_ih_i^{\text{sp}}-\frac c2 (h_i^{\text{sp}})^2=\frac c2 (h_i^{\text{sp}})^2,
\end{equation*}
and also the profit of the buyer:
\begin{equation*}
\Pi(\bm h^{\text{sp}};T)=K(\bm h^{\text{sp}})-F-\sum_{i=1}^N p_ih_i^{\text{sp}}.
\end{equation*}
Since $p_i=ch_i^{\text{sp}}$ (from \Cref{lem:IC}) and the social planner solution ensures (recall \Cref{eq:FOC sp}) $\nabla K(\bm h^{\text{sp}})=c\bm h^{\text{sp}}$, we have
\begin{equation*}
\Pi(\bm h^{\text{sp}};T)=K(\bm h^{\text{sp}})-F-(c\bm h^{\text{sp}})'\bm h^{\text{sp}}=K(\bm h^{\text{sp}})-\big (\nabla K(\bm h^{\text{sp}})\big )'\bm h^{\text{sp}}-F.
\end{equation*}
In the special case where all creators are uncorrelated (i.e., $\beta_i=0$ for all $i$), we have $K(\bm h)=\sum_{i=1}^N h_i$ (recall \Cref{lem:precision common factor}).\footnote{Or more generally, in any {Constant Returns to Scale} (CRS) regime where $K(a \bm h)=a K(\bm h)$ for any scalar $a>0$.} Euler's Homogeneous Function Theorem gives
\begin{equation}\label{eq:Euler theorem zero mean}
K(\bm h) = \sum_{i=1}^N h_i \cdot \frac{\partial K}{\partial h_i}(\bm h)=(\nabla K(\bm h))'\bm h,\quad \forall \bm h\in \mathbb R_{\ge 0}^N.
\end{equation}
Therefore, the buyer's profit is $\Pi=-F\le 0$ in this case, which violates \Cref{eq:IR} whenever $F+\underline \Pi>0$.
\end{proof}

\subsection{Two-Part Tariff for Efficiency and Participation (Proof of \Cref{obs:two part})}

\begin{proof}[Proof of \Cref{obs:two part}]
From the specific choice of unit price $\bm p=c\bm h^{\text{sp}}$, Incentive Compatibility \eqref{eq:IC} holds due to \Cref{lem:IC}.
Since $\bm \phi=(\phi_1,\phi_2,\ldots,\phi_N)'$ sums to 1, at the equilibrium effort $\bm h^{\text{sp}}$ the total payment to the creators is (recall the $B_i$ defined in \Cref{eq:subsidy two part}):
\begin{equation*}\begin{aligned}
&\quad \sum_{i=1}^N P_i(\bm h^{\text{sp}})=\sum_{i=1}^N \big (p_ih_i^{\text{sp}}+B_i\big )\\
&=\sum_{i=1}^N \left [ch_i^{\text{sp}}\times h_i^{\text{sp}}+\phi_i\alpha \mathcal S(\bm h^{\text{sp}})+\underline u_i-\frac c2 (h_i^{\text{sp}})^2\right ]\\
&=\alpha \mathcal S(\bm h^{\text{sp}})+\sum_{i=1}^N \left [\underline u_i+\frac c2 (h_i^{\text{sp}})^2\right ].
\end{aligned}\end{equation*}
Meanwhile, the transfer received from the buyer (recall the $L$ defined in \Cref{eq:transfer two part}) ensures
\begin{equation*}
T(\bm h^{\text{sp}})=K(\bm h^{\text{sp}})-(1-\alpha)\mathcal S(\bm h^{\text{sp}})-\underline \Pi.
\end{equation*}

The surplus of the intermediary (income minus operational cost minus payments) then satisfies
\begin{equation*}\begin{aligned}
T(\bm h^{\text{sp}})-F-\sum_{i=1}^N P_i(\bm h^{\text{sp}})&=K(\bm h^{\text{sp}})-(1-\alpha)\cdot \mathcal S(\bm h^{\text{sp}})-\underline \Pi - F -\alpha \mathcal S(\bm h^{\text{sp}})-\sum_{i=1}^N \left [\underline u_i+C(h_i^{\text{sp}})\right ] = 0,
\end{aligned}\end{equation*}
by definition of $\mathcal S(\bm h)$. This proves Budget Balance \eqref{eq:BB} condition.
We finally verify the Individual Rationality \eqref{eq:IR} condition. For each creator $i=1,2,\ldots,N$:
\begin{equation*}\begin{aligned}
U_i(\bm h^{\text{sp}};\bm P) = P_i(\bm h^{\text{sp}})-C(h_i^{\text{sp}}) &= \left [ch_i^{\text{sp}}\times h_i^{\text{sp}}+\phi_i\alpha \mathcal S(\bm h^{\text{sp}})+\underline u_i-\frac c2 (h_i^{\text{sp}})^2\right ] - C(h_i^{\text{sp}}) \\
&=\underline u_i+\phi_i \alpha \mathcal S(\bm h^{\text{sp}})\ge \underline u_i,
\end{aligned}\end{equation*}
where the last step is because $\mathcal S(\bm h^{\text{sp}})> 0$ (the condition of \Cref{obs:two part}).
The buyer's profit:
\begin{equation*}\begin{aligned}
\Pi(\bm h^{\text{sp}};T) = K(\bm h^{\text{sp}})-T(\bm h^{\text{sp}}) &= K(\bm h^{\text{sp}}) - \big( K(\bm h^{\text{sp}})-(1-\alpha) \mathcal S(\bm h^{\text{sp}})-\underline \Pi \big) \\
&=\underline \Pi + (1-\alpha) \mathcal S(\bm h^{\text{sp}})\ge \underline \Pi,
\end{aligned}\end{equation*}
where we again used the condition that $\mathcal S(\bm h^{\text{sp}})> 0$.
\end{proof}

\end{document}